\documentclass [a4paper,12pt]{article}
\usepackage[left=2cm,right=2cm,
    top=2cm,bottom=2cm,bindingoffset=0cm]{geometry}
\usepackage{graphicx}
\usepackage[small]{subfigure,epsfig}
\usepackage{indentfirst}
\usepackage{amsmath,latexsym,enumerate}
\usepackage{amssymb}
\usepackage{amsfonts}
\usepackage{cite}
\usepackage{amsthm}

\theoremstyle{plain} \newtheorem{theorem}{Theorem}
\numberwithin{equation}{section} \numberwithin{table}{section} \numberwithin{theorem}{section}

\begin{document}

\title{\textbf{Multi--particle dynamical systems and polynomials}}

\author{Maria V. Demina and Nikolai A. Kudryashov}

\date{Department of Applied Mathematics\\
National Research Nuclear University "MEPhI"\\
31 Kashirskoe Shosse, 115409, Moscow, \\ Russian Federation}

\maketitle

\begin{abstract}

Polynomial  dynamical systems describing interacting particles in the plane  are studied. A method replacing integration of a polynomial  multi--particle dynamical system by finding polynomial solutions of a partial differential equations is described.    The method enables one to integrate a wide class of polynomial multi--particle dynamical systems.
The general solutions of certain dynamical systems related to linear second--order partial differential equations are found. As a by-product of our results, new families of orthogonal polynomials are derived. Our approach is also applicable to dynamical systems that are not multi--particle
by their nature but that can be regarded as multi--particle (for example, the  Darboux--Halphen system and its generalizations). A wide class of two and three--particle polynomial dynamical systems is integrated.

\end{abstract}


\section{Introduction} \label{I:MP}

Integrating an ordinary differential equation is one of the major problems of analysis.
With the exception of certain classes of ordinary differential equations this problem is rather complicated. A. Cauchy
suggested to study ordinary differential equations within the framework of complex analysis, what allowed him to obtain significant
results on local existence and uniqueness of solutions. S.V. Kovalevskaya was one of the first mathematicians who noted a remarkable connection between explicit
integrability of an ordinary differential equation and the singularity structure of its general solution \cite{Kowalevski01, Kowalevski02}. Ideas of S.V. Kovalevskaya were extended and developed by E. Picard, G. Mittag--Leffler, P. Painlev\'{e}, B. Gambier, F.J. Bureau.

If the general solution of an  ordinary differential equation
does not have movable critical points, then it can be uniformized to fit the definition of a function (as a single valued mapping). Absence of movable critical points in the general solution of an ordinary differential equation is now called the Painlev\'{e} property in honor of French mathematician P. Painlev\'{e}.  It can be concluded
that  an ordinary differential equation possessing the  Painlev\'{e} property is integrable in known functions or itself gives rise to a new function. L. Fuchs and H. Poincar\'{e} suggested to look for new functions defined by ordinary differential equations. This problem can be solved in two steps. The first step is to make a classification of ordinary differential equations in a certain class, whose general solutions do not have movable critical points and, consequently, define functions. The second step consists in selection of those equations that possess the general solutions not expressible via known functions.

All linear equations generate functions. Nonlinear first--order algebraic ordinary differential equations give rise to only one new class of functions, the elliptic functions. A considerable contribution into the classification of second--order ordinary differential equations with the general solutions without movable critical points was done by P. Painlev\'{e}.

At the turn of the twentieth century  the research group headed by P. Painlev\'{e} performed the classification of second--order ordinary differential equations of the form
\begin{equation}
\begin{gathered}
\label{Introduction_Painleve}
w_{tt}=R(w_t,w,t),
\end{gathered}
\end{equation}
 where $R$ is rational in $w_t$ and $w$ \cite{Painleve01, Gambier01}. P. Painlev\'{e} and his colleagues found fifty canonical equations with the general solutions possessing  the Painlev\'{e} property.  Forty four equations can be integrated in terms of previously known functions  and six equations required the introduction of new special functions. Nowadays these equations are called the Painlev\'{e} equations and their general solutions are referred to as the Painlev\'{e} transcendents. The complete list of equations~\eqref{Introduction_Painleve} with the Painlev\'{e} property can be found in \cite{Ince01}.

 The classification programm did not finish at second--order equations. J. Chazy made the classification of third--order ordinary differential equations in the polynomial class  possessing the Painlev\'{e} property~\cite{Chazy01}. J. Chazy considered equations in the form
\begin{equation}
\begin{gathered}
\label{Chazy_class}
w_{ttt}=R(w_{tt},w_{t},w,t),
\end{gathered}
\end{equation}
where $R$ is a polynomial in $w_{tt}$, $w_{t}$, $w$. The work of J. Chazy was developed by F.J. Bureau, H. Exton, I. P. Martynov, C.M. Cosgrove, see~\cite{Cosgrove01} and references therein. The classification of ordinary differential equations is still going on. There exists a number of fourth and higher--order nonlinear  equations that are supposed to define new functions \cite{Airault01, Flaschka01, Kudryashov01, Kudryashov02}. However, this hypothesis is not proved yet \cite{Kudryashov02}.

The problem of integrating a system of nonlinear ordinary differential  equations, especially if one can not obtain a single equation satisfied by some dependent function expressing all other functions from given equations, is even more difficult. In this article we consider systems of ordinary differential equations that can be regarded as multi--particle in the sense that given equations describe dynamics of point particles in the plane. Our aim is to integrate certain classes of multi--particle dynamical systems.

Multi–-particle dynamical systems, such as collections of interacting point vortices
in the plane and on a sphere, have been attracting much attention during recent years.
The point vortex system, an elegant and visual model of fluid dynamics, is not
integrable in the case of four or more vortices with generic choice of circulations \cite{Aref05, Borisov04}.
That is why particular motions, including relative and absolute equilibria, collapse, and
scattering are of great importance. A well--known class of absolute equilibria involving
point vortices with equal in absolute value circulations is given by the roots of two neighbor Adler--Moser polynomials \cite{Bartman01}. It is a remarkable fact that the roots of the Adler--Moser polynomials
themselves provide solutions of another multi--particle system, the Airault--McKean--Moser
dynamical system, related to the Korteweg -- de Vries equation~\cite{Airault04, Adler01}.

Not long ago a method enabling one to find absolute and relative equilibrium configurations of point vortices
in the plane and on a sphere was introduced and developed \cite{Aref01, Oneil02, Demina25, Demina26, Demina27, Demina28, Demina29, Demina30}. The starting point of the method is a polynomial or a system
of polynomials with roots at the vortex positions. Further these polynomials are
shown to obey certain ordinary differential equations. The strength of the method lies in the fact that it works in both directions: from,
for example, point vortex relative equilibria to a differential equation and vice versa.

In this article we generalize the polynomial method to the case of polynomial  multi--particle dynamical systems.
In fact, we shall study  two problems.
The first problem consists in finding a dynamical system satisfied by the roots of a polynomial that obeys a partial differential
equation. The second problem is opposite to the first. Suppose that one originates with a polynomial multi--particle dynamical system.
The problem is to find a partial differential equation of degree less than the amount of particles such that the polynomial having  roots coinciding with
the particle positions in the plane satisfies this equation.

The first problem has been  considered previously (see \cite{Airault04, Choodnovsky01}, works by F. Calogero \cite{Calogero01, Calogero02} and references therein). However, most of the authors deal with the case of rational solutions. In other words, there have been subsequent studies devoted to dynamical systems satisfied by the poles of rational solutions of integrable partial differential equations \cite{Airault04, Choodnovsky01,Calogero01, Calogero02}. While it seems that the second problem yet has not  attracted any attention.

Using the polynomial method we find solutions of some interesting dynamical systems related to linear second--order partial differential equations. As a byproduct of our results we derive new families of orthogonal polynomials.

Our area of consideration includes not only dynamical systems multi--particle by their nature, but also systems that  can be regarded as multi--particle. We integrate a wide class of two and three--particle polynomial  dynamical systems. This class includes a number of systems interesting from physical point of view, such as the Euler's and Darboux--Halphen systems. The Euler's equations arise in rigid body dynamics. The Darboux--Halphen system finds applications in mathematical physics in relation to magnetic monopole dynamics, self--dual Einstein equations, topological field theory and in other fields of science \cite{DH, Darboux01, Halphen01, Halphen02, Gibbons01}.   Some  generalizations of these systems also have a number of applications \cite{DH, Ablowitz01}. For a class of two and three--particle polynomial  dynamical systems including the aforementioned examples and their generalizations we present the cases, when such systems can be regarded as integrable in the sense that one can obtain their general solutions. For some other cases we give a number of exact  elliptic solutions.  In fact, we find all second--order elliptic solutions of a third--order differential equation arisen after application of the polynomial method.

The problem of finding and classifying exact solutions of nonlinear ordinary and partial differential equations is of great theoretical and practical importance.
In the past decades there has been a significant progress in the development of
these methods. Group theoretical techniques provide  exact solutions of equations possessing certain symmetries.   Several powerful methods, such as the Hirota
bilinear method, algorithms based on Darboux transformations and Wronskian representations, are designed mainly for partial differential equations integrable by inverse scattering method. An enormous number of methods deal with traveling wave solutions. Let us name only a few: the tanh--function method, the exponential method, the Jacobi elliptic--function method and their various extensions and modifications. Most of these methods use fixed expressions of unknown solutions. Consequently, all solutions lying outside supposed representations  are lost. Along with this, such methods usually give the same solutions, but written in a different way. Consequently, this class of methods cannot be used if one wishes to perform a classification of exact solutions with given properties. In this article we shall use a method of finding exact elliptic solutions, which is free from these disadvantages.

Our approach is based on Mittag--Leffler's expansions of a meromorphic functions. This method allows one to find explicitly any elliptic solution of an algebraic ordinary differential equation. Consequently, the method can be used if one needs to classify elliptic solutions.

This article is organized as follows. In section \ref{Multi_particle_gen} we present our method and study the problem of finding dynamical systems related to polynomial solutions of partial differential equations. Section \ref{Symmetric_DS} is devoted to the problem of constructing a partial differential equation possessing polynomial solutions with the roots obeying a given polynomial multi--particle dynamical system. In section \ref{Nparticle_linear} we consider multi--particle dynamical systems related to linear second--order partial differential equations.  In section \ref{Sec:MP23} we study a class of  two and three--particle polynomial dynamical systems, including physically meaningful ones.


\section{Method applied} \label{Multi_particle_gen}

We begin with some preliminary remarks. Consider a polynomial in $z$ with time--dependent coefficients:
\begin{equation}
\begin{gathered}
\label{Polynomial_general_coeff}
p(z,t)=c_0(t)z^M+\sum_{j=0}^{M-1}c_{M-j}(t)z^j,\quad c_0(t)\not\equiv 0.
\end{gathered}
\end{equation}
We suppose that the polynomial $p(z,t)$ does not have multiple roots. This assumption gives a representation
\begin{equation}
\begin{gathered}
\label{Polynomial_general_roots}
p(z,t)=c_0(t)\prod_{j=1}^{M}\left\{z-a_j(t)\right\},
\end{gathered}
\end{equation}
where $a_j(t)\not\equiv a_k(t)$, $j\neq k$. Unless otherwise is stated, let $z$ be a complex variable and let $a_j(t)$, $c_{j}(t)$ be complex--valued functions of a complex variable $t$. If we are interested in physically relevant solutions, then we shall restrict ourselves with real values of $t$. Calculating logarithmic derivatives of the polynomial $p(z,t)$ yields relations
\begin{equation}
\begin{gathered}
\label{Logarithmic derivatives_z}
\frac{p_z}{p}=\sum_{j=1}^{M}\frac{1}{z-a_j},\quad \frac{p_{zz}}{p}-\frac{p_{z}^2}{p^2}=\sum_{j=1}^{M}\frac{(-1)}{\left\{z-a_j\right\}^2},
\quad \frac{d^n}{dz^n}\log p =\sum_{j=1}^{M}\frac{(-1)^{n+1}(n-1)!}{\left\{z-a_j\right\}^{n}}
\end{gathered}
\end{equation}
and
\begin{equation}
\begin{gathered}
\label{Logarithmic derivatives_tz}
\frac{p_t}{p}=\sum_{j=1}^{M}\frac{\left(-a_{j,t}\right)}{z-a_j}+\frac{c_{0,t}}{c_0},\quad \frac{p_{tz}}{p}-\frac{p_tp_z}{p^2}=\sum_{j=1}^{M}\frac{\left(-a_{j,t}^2\right)}{\left\{z-a_j\right\}^2},\\
\frac{p_{tt}}{p}-\frac{p_{t}^2}{p^2}= \sum_{j=1}^{M}\frac{\left(-a_{j,tt}\right)}{z-a_j}+\sum_{j=1}^{M}\frac{\left(-a_{j,t}^2\right)}{\left\{z-a_j\right\}^2}+\frac{c_{0,tt}}{c_0}-\frac{c_{0,t}^2}{c_0^2}
\end{gathered}
\end{equation}
For further convenience we introduce notation
\begin{equation}
\begin{gathered}
\label{Polynomials_Notation}
L_m=\sum_{j=1,j\neq k}^{M}\frac{1}{\left\{a_k-a_j\right\}^m},\quad
G_m=\sum_{j=1,j\neq k}^{M}\frac{a_{j,t}}{\left\{a_k-a_j\right\}^m},\quad m\in\mathbb{N}.
\end{gathered}
\end{equation}
The derivatives of the function $a_k(t)$ and the quantities $L_m$, $G_m$ can be expressed through the derivatives of the polynomial $p(z,t)$. Let us consider in detail the derivation of $a_{k,t}$, $a_{k,tt}$, $L_1$. Multiplying the first relation in \eqref{Logarithmic derivatives_tz}
by $p$ and calculating the limit $z\rightarrow a_k$ we easily get $a_{k,t}=-\{p_t/p_z\}_{z=a_k}$.  In order to find the limit in the right--hand
side of the relation
we have applied the l'H\^{o}pital's rule. Since the polynomial $p(z,t)$ does not have multiple roots, the following condition is valid $\{p_z\}_{z=a_k}\not\equiv 0$. Further, taking the second relation  in \eqref{Logarithmic derivatives_tz}
we find
\begin{equation}
\begin{gathered}
\label{Nparticle_calculations1}
a_{k,tt}\{p_z\}_{z=a_k}=\lim_{z\rightarrow a_k}\frac{(p_t^2-p_{tt}p)\left\{z-a_k\right\}^2-a_{k,t}^2p}{\left\{z-a_k\right\}^2p}=\left\{\frac{2p_tp_z p_{tz}-p_t^2p_{zz}-p_z^2p_{tt}}{p_z^2}\right\}_{z=a_k}.
\end{gathered}
\end{equation}
Let us calculate the quantity $L_1$. For this aim we take the first relation in \eqref{Logarithmic derivatives_z}, subtract from both sides the expression $\{z-a_k\}^{-1}$ and use the  l'H\^{o}pital's rule to obtain $L_1=\{p_{zz}/(2p_z)\}_{z=a_k}$. Analogously one can calculate higher--order derivatives of the function $a_k$, if necessary, and quantities $L_m$, $G_m$. Let us write down those that we shall use later
\begin{equation}
\begin{gathered}
\label{Nparticle_relations1}
a_{k,t}=-\frac{p_t}{p_z},\quad a_{k,tt}=\frac{2p_tp_{tz}}{p_z^2}-\frac{p_t^2p_{zz}}{p_z^3}-\frac{p_{tt}}{p_z},\quad G_1=\frac{p_tp_{zz}}{2p_z^2}-\frac{p_{tz}}{p_z}+\frac{c_{0,t}}{c_0},\\
L_1=\frac{p_{zz}}{2p_z},\quad L_2=\frac{p_{zz}^2}{4p_z^2}-\frac{p_{zzz}}{3p_z},\quad L_3=\frac{p_{zzzz}}{8p_z}-\frac{p_{zz}p_{zzz}}{4p_z^2}
+\frac{p_{zz}^3}{8p_z^3}.\hfill
\end{gathered}
\end{equation}
In expressions \eqref{Nparticle_relations1} all the derivatives of the polynomial $p(z,t)$ are taken at its root $a_k$. Obtained relations are of rather general character. They are valid for any polynomial with simple roots.

Further let us note that any derivative of the polynomial $p(z,t)$ at its root $a_k$ divided by $(p_z)_{z=a_k}$ can be expressed through the quantities $a_{k,t}$, $a_{k,tt}$, $L_m$, $G_m$ (and their analogues arising in the case one wishes to find $p_{ttt}$, $p_{zzt}$ etc). Such relations can be obtained expressing step by step the corresponding derivatives from the equalities \eqref{Nparticle_relations1}. For further purposes let us give several of them
\begin{equation}
\begin{gathered}
\label{Nparticle_relations2}
p_t=-a_{k,t}p_z,\, p_{tz}=-(G_1+a_{k,t}L_1-\{\log c_0\}_t)p_z,\, p_{tt}=(2a_{k,t}[G_1-\{\log c_0\}_t]\hfill \\
-a_{k,tt})p_z,\quad
p_{zz}=2L_1p_z,\quad p_{zzz}=3(L_1^2-L_2)p_z,\quad p_{zzzz}=4(2L_3-3L_1L_2+L_1^3)p_z.
\end{gathered}
\end{equation}
Here again all the derivatives of the polynomial $p(z,t)$ are taken at its root $a_k$. In addition we have
\begin{equation}
\begin{gathered}
\label{Nparticle_relations3}
\{p_z\}_{z=a_k}=\prod_{j=1,j
\neq k}^{M}\{a_k-a_j\}
\end{gathered}
\end{equation}

Now let us describe the polynomial method. In this section we shall mainly address  the first problem from those stated in the introduction.
Consider a  partial differential equation
\begin{equation}
\begin{gathered}
\label{Nparticle_PDE_llinear}
E(t,z,p_z,p_t,p_{tz},p_{tt},p_{zz},\ldots)=0,
\end{gathered}
\end{equation}
where  $E$ is a polynomial in $z$, $p(z,t)$ and its derivatives. Here and in what follows we suppose that all the coefficient functions of partial differential equations and dynamical systems are well-behaved functions of the parameter $t$ (i.e. entire or meromorphic).

Suppose that a polynomial $p(z,t)$ with simple roots solves this equation; then substituting $z=a_k$ and relations of the form \eqref{Nparticle_relations2} into equation \eqref{Nparticle_PDE_llinear}, one obtains dynamical equations satisfied by the roots of the polynomial $p(z,t)$. Thus the $M$ zeros of the polynomial $p(z,t)$ are interpreted as the coordinates of $M$ point particles.

The strength of the polynomial method lies in the fact that it enables one to restore original equation \eqref{Nparticle_PDE_llinear}. Indeed, introducing a polynomial $p(z,t)$ with the roots at particle positions and substituting equalities of the form \eqref{Nparticle_relations1} into the equations of motion and getting rid of the denominators, one arrives at the following relations
\begin{equation}
\begin{gathered}
\label{Nparticle_calculations2}
\left\{F(t,z,p_z,p_t,p_{tz},p_{tt},p_{zz},\ldots)\right\}_{z=a_k}=0,\quad k=1 \ldots M,
\end{gathered}
\end{equation}
where $F$ is a polynomial in $z$, $p$, its derivatives and, consequently, a polynomial in $z$. This polynomial possesses $M$ roots $a_1$, $\ldots$, $a_M$. Thus we conclude that
\begin{equation}
\begin{gathered}
\label{Nparticle_calculations3}
F(t,z,p_z,p_t,p_{tz},p_{tt},p_{zz},\ldots)-P(z,t)p=0.
\end{gathered}
\end{equation}
In this relation $P(z,t)$ is a polynomial in $z$ such that $\deg P=\deg F- M$. If $\deg F<M$ then $P(z,t)\equiv0$. In nonlinear cases the polynomial $P$ may depend on $p(z,t)$ and its derivatives. At this step  equation \eqref{Nparticle_calculations3} may appear to be more general than the original equation  \eqref{Nparticle_PDE_llinear}. In other words, equations \eqref{Nparticle_PDE_llinear}, \eqref{Nparticle_calculations3} coincide accurate to the polynomial $P(z,t)$. If we wish to identify the polynomial $P(z,t)$ and to establish a correspondence between the polynomial solution $p(z,t)$ of equation \eqref{Nparticle_PDE_llinear} and a dynamical system, then obtained dynamical equations may need additional constrains. In order to find them one can, for example, repeat the described procedure taking the differential consequences of original equation \eqref{Nparticle_PDE_llinear}. The complete list of this constrains can be derived in the following way. Making the substitutions $p_z=up$, $p_t=vp$ into equation \eqref{Nparticle_PDE_llinear} gives a partial differential equation with dependent variables $p(z,t)$, $u(z,t)$ and $v(z,t)$ satisfying in addition the equation $u_t=v_z$. If $p(z,t)$ is a polynomial with simple roots, then the functions $u(z,t)$, $v(z,t)$ are given by
\begin{equation}
\begin{gathered}
\label{Nparticle_PDE_v_w}
u(z,t)=\sum_{j=1}^{M}\frac{1}{z-a_j},\quad v(z,t)=\sum_{j=1}^{M}\frac{(-a_{j,t})}{z-a_j}+\frac{c_{0,t}}{c_0}.
\end{gathered}
\end{equation}
Calculating the generalized Laurent series in a neighborhood of the pole $a_k$ for the functions $u(z,t)$, $v(z,t)$  yields
\begin{equation}
\begin{gathered}
\label{Nparticle_PDE_v_w}
u(z,t)=\frac{1}{z-a_k}+\sum_{m=0}^{\infty}(-1)^mL_{m+1}\{z-a_k\}^m,\quad z\rightarrow a_k\hfill \\
v(z,t)=\frac{(-a_{k,t})}{z-a_k}+\frac{c_{0,t}}{c_0}+\sum_{m=0}^{\infty}(-1)^{m+1}G_{m+1}\{z-a_k\}^m,\quad z\rightarrow a_k.
\end{gathered}
\end{equation}
The  generalized Laurent series in a neighborhood of infinity are the following
\begin{equation}
\begin{gathered}
\label{Nparticle_PDE_v_w_infinity}
u(z,t)=\frac{M}{z}+\sum_{m=1}^{\infty}\left\{\sum_{j=1}^{M}a_j^m\right\}z^{-m-1},\quad z\rightarrow \infty \hfill\\
v(z,t)=\frac{c_{0,t}}{c_0}-\sum_{m=0}^{\infty}\left\{\sum_{j=1}^{M}a_{j,t}a_j^m\right\}z^{-m-1},\quad z\rightarrow \infty.\hfill
\end{gathered}
\end{equation}
Substituting  series  \eqref{Nparticle_PDE_v_w}, \eqref{Nparticle_PDE_v_w_infinity},
\begin{equation}
\begin{gathered}
\label{Nparticle_PDE_p_series}
p(z,t)=\{p_z\}_{z=a_k}(z-a_k)+\frac{\{p_{zz}\}_{z=a_k}}{2}(z-a_k)^2+\ldots+(z-a_k)^M,\quad z\rightarrow a_k,
\end{gathered}
\end{equation}
 and  relation \eqref{Polynomial_general_coeff}, which is in fact the generalized Laurent series of the polynomial $p(z,t)$ in a neighborhood of infinity, into the partial differential equation relating $u$, $v$, $p$ and setting to zero the corresponding  coefficients at negative and zero powers of $\{z-a_k\}$, $k=1$, $\ldots$, $M$ and $z^{-1}$  gives the desired  system. Indeed, the left--hand side of the equation relating $u$, $v$, $p$ is a rational function without poles provided that this system is satisfied. From the Liouville theorem it immediately follows that such a function is a constant, which equals zero since the coefficients at  zero powers of $\{z-a_k\}$, $k=1$, $\ldots$, $M$, $z$ vanish.   Note that one may take only one correlation at the zero level $\{z-a_k\}^0$, $k=1$, $\ldots$, $M$, $z^{0}$. The coefficients $\{c_m(t)\}$ of the polynomial $p(z,t)$ are expressible via the dynamical variables $a_1(t)$, $\ldots$, $a_m(t)$ as follows
 \begin{equation}
\begin{gathered}
\label{Sym_Pol_Sym_Func_Coeff_non_monic}
c_m=(-1)^mS_mc_0,\quad m=1 \ldots M,
\end{gathered}
\end{equation}
where $S_m$ are the elementary symmetric functions, see formulae \eqref{Sym_Pol_Elem_Sym_Func} of section \ref{Symmetric_DS}.

As soon as the dynamical system is supplemented with additional constrains, then they can be used to identify the polynomial $P(z,t)$ in \eqref{Nparticle_calculations3}. For this aim one rewrites these constrainers via $\{p_z\}_{z=a_k}$, $\{p_{t}\}_{z=a_k}$, etc., differentiates  equation \eqref{Nparticle_calculations3} with respect to $z$, and substitutes the corresponding derivatives expressed from the additional constrains into the resulting equation (see example below).

For more details on constructing a partial differential equation related to a polynomial multi--particle dynamical system see section \ref{Symmetric_DS}.

Interestingly, the similar algorithm can be used to relate rational solutions of a partial differential equation and a dynamical system obeyed by the poles of its rational solutions.


Let us consider several examples. It is known that the heat equation
\begin{equation}
\begin{gathered}
\label{Nparticle_heat_PDE}
p_t-p_{zz}=0.
\end{gathered}
\end{equation}
possesses an infinite series of monic polynomial solutions, the so--called heat polynomials. Substituting relations \eqref{Nparticle_relations2} into equation \eqref{Nparticle_heat_PDE} we immediately get $a_{k,t}=-2L_1$ or explicitly
\begin{equation}
\begin{gathered}
\label{Nparticle_heat_DS}
a_{k,t}=-2\sum_{j=1,j\neq k}^{M}\frac{1}{a_k-a_j},\quad k=1 \ldots M.
\end{gathered}
\end{equation}
If a starting point is the system \eqref{Nparticle_heat_DS} then introducing a monic polynomial $p(z,t)$ with roots at the particle positions and using equalities \eqref{Nparticle_relations1} we obtain
\begin{equation}
\begin{gathered}
\label{Nparticle_heat_Calculations1}
\{p_t-p_{zz}\}_{z=a_k}=0.
\end{gathered}
\end{equation}
The following inequality $\deg(p_t-p_{zz})$ $<$ $M$ is valid whenever $p(z,t)$ is a monic polynomial. Consequently, the polynomial $p(z,t)$ satisfies the heat equation. In addition, we see that system \eqref{Nparticle_heat_DS} does not need any additional constraint (with the only exception for initial conditions). If all the functions $a_1(t)$, $\ldots$, $a_M(t)$ are real then the corresponding system describes dynamics of identical point vortices on a line.

As an illustrative nonlinear example let us take the following bilinear partial differential equation
\begin{equation}
\begin{gathered}
\label{Nparticle_KdV_bilinear}
pp_{tz}-p_zp_t+pp_{zzzz}-4p_zp_{zzz}+3p_{zz}^2=0.
\end{gathered}
\end{equation}
Each polynomial from the sequence of the Adler--Moser polynomials is a monic polynomial solution of this equation. For some properties of the Adler--Moser polynomials see \cite{Airault04, Adler01, Demina31, Demina32}. Substituting relations \eqref{Nparticle_relations2} into equation \eqref{Nparticle_KdV_bilinear} we obtain the following multi--particle dynamical equations
\begin{equation}
\begin{gathered}
\label{Nparticle_KdV_DS}
a_{k,t}=-12\sum_{j=1,j\neq k}^{M}\frac{1}{\left\{a_k-a_j\right\}^2},\quad k=1 \ldots M.
\end{gathered}
\end{equation}
Our goal is to establish a correspondence between monic polynomials with simple roots that satisfy equation \eqref{Nparticle_KdV_bilinear} and a dynamical system. Substituting $p_z=up$, $p_t=vp$ into equation \eqref{Nparticle_KdV_bilinear}, we obtain
\begin{equation}
\begin{gathered}
\label{Nparticle_KdV_bilinear_after_tr}
u_t+6u_z^2+u_{zzz}=0.
\end{gathered}
\end{equation}
Note that differentiating this equation with respect to $z$ and introducing the new variable $\tilde{u}=2u_z$ yields the Korteweg -- de Vries equation
\begin{equation}
\begin{gathered}
\label{Nparticle_KdV}
\tilde{u}_t+6\tilde{u}\tilde{u}_z+\tilde{u}_{zzz}=0.
\end{gathered}
\end{equation}
Substituting series \eqref{Nparticle_PDE_v_w} into equation \eqref{Nparticle_KdV_bilinear_after_tr} and setting to zero coefficients at $\{z-a_k\}^{-2}$ and $\{z-a_k\}^{-1}$ gives \eqref{Nparticle_KdV_DS} and  constrains of the form
\begin{equation}
\begin{gathered}
\label{Nparticle_KdV_DS_constrains}
\sum_{j=1,j\neq k}^{M}\frac{1}{\left\{a_k-a_j\right\}^3}=0,\quad k=1 \ldots M.
\end{gathered}
\end{equation}
An additional equation at the level $z^{0}$ is automatically satisfied. Originally dynamical system \eqref{Nparticle_KdV_DS}, \eqref{Nparticle_KdV_DS_constrains} was found by Airault, McKean, and Moser \cite{Airault04}. Note that equations \eqref{Nparticle_KdV_DS}, \eqref{Nparticle_KdV_DS_constrains} are compatible provided that $M$ is a triangular number \cite{Airault04}.

Now let us construct bilinear equation \eqref{Nparticle_KdV_bilinear} originating from system \eqref{Nparticle_KdV_DS}, \eqref{Nparticle_KdV_DS_constrains}. Introducing a monic polynomial $p(z,t)$ with roots at the particle positions and making use of relations \eqref{Nparticle_relations1} we get
\begin{equation}
\begin{gathered}
\label{Nparticle_KdV_Calculations1}
\left\{p_tp_z+4p_zp_{zzz}-3p_{zz}^2\right\}_{z=a_k}=0,\quad k=1 \ldots M;\hfill\\
 \left\{p_z^2p_{zzzz}-2p_zp_{zz}p_{zzz}+p_{zz}^3\right\}_{z=a_k}=0,\quad k=1 \ldots M.
\end{gathered}
\end{equation}
From the first set of these relations we obtain
\begin{equation}
\begin{gathered}
\label{Nparticle_KdV_Calculations2}
p_tp_z+4p_zp_{zzz}-3p_{zz}^2-P(z,t)p=0,
\end{gathered}
\end{equation}
where $P(z,t)$ is a polynomial in $z$ of degree $M-2$. Differentiating this equation with respect to $z$ and setting $z=a_k$ yields
\begin{equation}
\begin{gathered}
\label{Nparticle_KdV_Calculations2}
\left\{p_{tz}p_z+p_tp_{zz}+4p_zp_{zzzz}-2p_{zz}p_{zzz}-P(z,t)p_z\right\}_{z=a_k}=0,
\end{gathered}
\end{equation}
Our goal is to create  an expression with a common multiplier $p_z$. Consequently, we express $p_t$
from  the first set of relations in \eqref{Nparticle_KdV_Calculations1} and find $p_{zz}^3$ from the second set of relations in \eqref{Nparticle_KdV_Calculations1} and substitute the results into expressions \eqref{Nparticle_KdV_Calculations2} to obtain
\begin{equation}
\begin{gathered}
\label{Nparticle_KdV_Calculations3}
\left\{[p_{tz}+p_{zzzz}-P(z,t)]p_z\right\}_{z=a_k}=0,
\end{gathered}
\end{equation}
The polynomial $p_{tz}+p_{zzzz}-P(z,t)$ is of degree $M-2$ and possesses $M$ roots $a_1$, $\ldots$, $a_M$. Thus this polynomial identically equals zero and we conclude that
\begin{equation}
\begin{gathered}
\label{Nparticle_KdV_Calculations4}
P(z,t)=p_{tz}+p_{zzzz}.
\end{gathered}
\end{equation}
This completes the derivation of equation \eqref{Nparticle_KdV_bilinear}.

Finishing this section let us note that all our constructions are valid provided that $a_j(t)\not\equiv a_k(t)$, $j\neq k$. However, we do not exclude the cases when there exists an isolated point $t=t_0$ such that $a_j(t_0)=a_k(t_0)$. This coincidence   gives rise to collisions of particles. In order to derive dynamical systems we perform the polynomial method in domains of the complex plane $t$, where  $a_j(t)\neq a_k(t)$ and further we use the principle establishing uniqueness of analytic continuation.

\section{Polynomial multi--particle dynamical systems} \label{Symmetric_DS}

In this section we shall originate with a multi--particle dynamical system and study the problem of finding a partial differential equation of degree less than the amount of particles such that the monic polynomial having  roots coinciding with
the particle positions in the plane satisfies this equation. Note that if a starting point is a multi--particle dynamical system, then we do not need to introduce non--monic polynomials.  Let us consider the following dynamical system
\begin{equation}
\begin{gathered}
\label{Sym_Pol_DS}
R(a_{k,tt},a_{k,t},a_k;a_1,\ldots,a_{k-1},a_{k+1},\ldots,a_M)=0,\quad k=1 \ldots M,
\end{gathered}
\end{equation}
where the function $R$ is a polynomial of its arguments with, possibly, $t$--dependent coefficients. In addition suppose that $R$ is symmetric with respect to the variables $a_1$, $\ldots$, $a_{k-1}$, $a_{k+1}$, $\ldots$, $a_M$. In the case $a_{k}\not \equiv a_{j}$, $k\neq j$ the system \eqref{Sym_Pol_DS} can be regarded as a multi--particle dynamical system  and the complex--valued functions $a_1(t)$, $\ldots$, $a_M(t)$ can be interpreted as particle positions in the plane. In what follows we shall call such a system polynomial multi--particle dynamical system. Note that we restrict ourselves we the first--order and second--order dynamical systems, since such systems are of great practical importance. While the polynomial method is applicable to polynomial multi--particle dynamical systems of arbitrary order.

It is known that any symmetric polynomial of $M-1$ variables  $a_1$, $\ldots$, $a_{k-1}$, $a_{k+1}$, $\ldots$, $a_M$ can be represented as the polynomial in the following symmetric functions
\begin{equation}
\begin{gathered}
\label{Sym_Pol_Sym_Func_Rest}
s_m^{\,'}=\sum_{j=1,\,j\neq k}^{M}a_j^m,\quad m\in\mathbb{N^{+}}
\end{gathered}
\end{equation}
This representation is unique and involves finite amount of these functions. Let us introduce a monic polynomial $p(z,t)$ with roots at the particle positions, see \eqref{Polynomial_general_coeff}, \eqref{Polynomial_general_roots}
 at $c_M(t)\equiv1$. Coefficients $c_1$, $\ldots$, $c_M$ of the polynomial $p(z,t)$  are symmetric polynomials with respect to the variables $a_1$, $\ldots$, $a_M$. Indeed,
\begin{equation}
\begin{gathered}
\label{Sym_Pol_Sym_Func_Coeff}
c_m=(-1)^mS_m,\quad m=1 \ldots M,
\end{gathered}
\end{equation}
where the  elementary symmetric functions $S_m$ are given by
\begin{equation}
\begin{gathered}
\label{Sym_Pol_Elem_Sym_Func}
S_1=a_1+a_2+\ldots +a_M,\quad S_2=a_1a_2+a_1a_3+\ldots,\\
 S_3=a_1a_2a_3+a_1a_2a_4+\ldots,\quad S_M=a_1a_2\ldots a_M \hfill
\end{gathered}
\end{equation}
In other words $S_m$ is the sum of all $C_M^m$ products, containing $m$ factors $a_j$ with distinct indices each. Equalities \eqref{Sym_Pol_Sym_Func_Coeff} should be replaced by \eqref{Sym_Pol_Sym_Func_Coeff_non_monic} whenever one wishes to consider the non--monic case.

It can be easily proved by induction that all the elementary symmetric functions $S_1$, $\ldots$, $S_M$ can be expressed via $a_k$ and the derivatives
\begin{equation}
\begin{gathered}
\label{Sym_Pol_Calc1}
\left\{\frac{\partial\, p}{\partial z}\right\}_{z=a_k},\left\{\frac{\partial^2\, p}{\partial z^2}\right\}_{z=a_k},\ldots,\quad \left\{\frac{\partial^{M-1}\, p}{\partial z^{M-1}}\right\}_{z=a_k}.
\end{gathered}
\end{equation}
Calculating the $z$--derivatives of the polynomial $p(z,t)$ and setting $z=a_k$, we get
\begin{equation}
\begin{gathered}
\label{Sym_Pol_Calc2}
 \left\{\frac{\partial^{m}\, p}{\partial z^{m}}\right\}_{z=a_k}=\frac{M!}{(M-m)!}a_k^{M-m}+\frac{(M-1)!}{(M-m-1)!}c_{1}a_k^{M-m-1}+\ldots\\
 +m!c_{M-m}, m=1 \ldots M.
\end{gathered}
\end{equation}
Relation \eqref{Sym_Pol_Calc2} at $m=M-1$ can be solved with respect to $c_{1}$. This gives
\begin{equation}
\begin{gathered}
\label{Sym_Pol_c_M1}
c_{1}=\frac{1}{(M-1)!} \left\{\frac{\partial^{M-1}\, p}{\partial z^{M-1}}\right\}_{z=a_k}-Ma_k.
\end{gathered}
\end{equation}
Analogously, solving relation \eqref{Sym_Pol_Calc2} at $m=M-2$ with respect to $c_{2}$ yields
\begin{equation}
\begin{gathered}
\label{Sym_Pol_c_M2}
c_{2}=\frac{1}{(M-2)!} \left\{\frac{\partial^{M-2}\, p}{\partial z^{M-2}}\right\}_{z=a_k}-\frac{M(M-1)}{2}a_k^2-(M-1)a_kc_{1}.
\end{gathered}
\end{equation}
Further, we substitute expression \eqref{Sym_Pol_c_M1} into equality \eqref{Sym_Pol_c_M2}. We solve relation \eqref{Sym_Pol_Calc2} at $m=M-l$ with respect to $c_{l}$. The coefficient $c_M$ is given by
 \begin{equation}
\begin{gathered}
\label{Sym_Pol_c_M3}
c_{M}=-\left\{a_k^M+c_1 a_k^{M-1}+\ldots+c_{M-1}a_k\right\}.
\end{gathered}
\end{equation}
With the help of expressions \eqref{Sym_Pol_Sym_Func_Coeff}, \eqref{Sym_Pol_Calc2}, \eqref{Sym_Pol_c_M3} we can calculate all the functions $S_1$, $\ldots$, $S_M$. As soon as these functions are known, we use the Newton formulae
\begin{equation}
\begin{gathered}
\label{Sym_Pol_Newton_Formulae}
s_m-s_{m-1}S_1+s_{m-2}S_2-\ldots +(-1)^{m-1}s_1S_{m-1}+(-1)^mmS_m=0,\quad 1\leq m\leq M;\hfill\\
s_m-s_{m-1}S_1+s_{m-2}S_2-\ldots +(-1)^{M}s_{m-M}S_{M}=0,\quad m>M\hfill
\end{gathered}
\end{equation}
to obtain the  symmetric functions
\begin{equation}
\begin{gathered}
\label{Sym_Pol_Sym_Func_small}
s_m=\sum_{j=1}^{M}a_j^m,\quad m\in\mathbb{N}.
\end{gathered}
\end{equation}
Consequently, the  symmetric functions $s_m^{\,'}=s_m-a_k^m$ of $M-1$ variables $a_1$, $\ldots$, $a_{k-1}$, $a_{k+1}$, $\ldots$, $a_M$ are polynomially expressible via $a_k$ and the derivatives \eqref{Sym_Pol_Calc2}. As we have already mentioned dynamical equations \eqref{Sym_Pol_DS} can be rewritten in terms of symmetric polynomials~$s_m^{\,'}$:
 \begin{equation}
\begin{gathered}
\label{Sym_Pol_DS_Rewritten}
\tilde{R}(a_{k,tt},a_{k,t},a_k;\{s_m^{\,'}\})=0,\quad k=1 \ldots M,
\end{gathered}
\end{equation}
 Substituting expressions of the form \eqref{Sym_Pol_c_M1}, \eqref{Sym_Pol_c_M2} and relations for $a_{k,t}$, $a_{k,tt}$, see \eqref{Nparticle_calculations1}, into the resulting dynamical equations \eqref{Sym_Pol_DS_Rewritten} and getting read of the denominators, we obtain the identities
\begin{equation}
\begin{gathered}
\label{Sym_Pol_Calc3}
\{F(t,z,p,p_z,p_t,\ldots)\}_{z=a_k}=0,\quad k=1 \ldots M,
\end{gathered}
\end{equation}
where $F$ is a polynomial in $z$, $p(z,t)$, its derivatives and, consequently, a polynomial in $z$ with M roots $a_1$,  $\ldots$, $a_M$. As a result we get the following partial differential equation
\begin{equation}
\begin{gathered}
\label{Sym_Pol_Calc4}
F(t,z,p,p_z,p_t,\ldots)-P(z,t)p=0
\end{gathered}
\end{equation}
with $P$ being a polynomial in $z$ of degree: $\deg P=\deg F-M$. The converse result is also valid. Suppose that we originate with a partial differential equation \eqref{Sym_Pol_Calc4} and its polynomial solution $p(z,t)$, see \eqref{Polynomial_general_roots}
 with $c_M(t)\equiv1$. Setting $z=a_k$ in equation \eqref{Sym_Pol_Calc4}, we  substitute relations \eqref{Nparticle_relations2} (also see expressions  \eqref{Sym_Pol_Sym_Func_Coeff}, \eqref{Sym_Pol_Calc2}) into the resulting equality. This gives the symmetric dynamical system.

The same approach is applicable to polynomial dynamical systems depending symmetrically not only on the variables $a_1$, $\ldots$, $a_{k-1}$, $a_{k+1}$, $\ldots$, $a_M$ but also on the variables $a_{1,t}$, $\ldots$, $a_{k-1,t}$, $a_{k+1,t}$, $\ldots$, $a_{M,t}$ in such a way that the system can be rewritten in the form
\begin{equation}
\begin{gathered}
\label{Sym_Pol_DS_Rewritten_new}
\tilde{R}(a_{k,tt},a_{k,t},a_k;\{s_m^{\,'}\},\{s_{m,t}^{\,'}\})=0,\quad k=1 \ldots M,
\end{gathered}
\end{equation}
In order to express the functions $s_{m,t}^{\,'}$ via $a_k$, $a_{k,t}$ and the derivatives
\begin{equation}
\begin{gathered}
\label{Sym_Pol_Calc1_t}
\left\{\frac{\partial^2\, p}{\partial z \partial t}\right\}_{z=a_k},\left\{\frac{\partial^3\, p}{\partial z^2 \partial t}\right\}_{z=a_k},\ldots,\quad \left\{\frac{\partial^{M}\, p}{\partial z^{M-1}\partial t}\right\}_{z=a_k}
\end{gathered}
\end{equation}
we consider the polynomial $p_t$ and calculate its $z$ derivatives at the point $z=a_k$
\begin{equation}
\begin{gathered}
\label{Sym_Pol_Calc2_t}
 \left\{\frac{\partial^{m+1}\, p}{\partial z^{m}\partial t}\right\}_{z=a_k}=\frac{(M-1)!}{(M-m-1)!}c_{1,t}a_k^{M-m-1}+\ldots+m!c_{M-m,t},\quad  m=0\ldots M.
\end{gathered}
\end{equation}
Further, we step by step express the quantities $c_{1,t}$, $\ldots$, $c_{M,t}$ from these relations and differentiate the Newton formulae \eqref{Sym_Pol_Newton_Formulae} to find $s_{m,t}$ and $s_{m,t}^{\,'}=s_{m,t}-ma_k^{m-1}a_{k,t}$.

Thus, the polynomial method enables one to place the study of  polynomial multi--particle dynamical systems in the framework of the theory of partial differential equations. For a wide class of symmetric polynomial dynamical systems  this approach yields only one ordinary differential equation for a certain coefficient of the polynomial $p(z,t)$, see section \ref{Sec:MP23}. In this case the problem of integrating a symmetric polynomial dynamical system reduces to the problem of solving one ordinary differential equation and an $M$-th order algebraic equation.

If a polynomial dynamical system with the dependent variables  $\xi=(\xi_1,\ldots,\xi_M)^{T}$ is not of multi--particle type, then
one may look for an invertible transformation $\xi=B(a)$ making the system in variables $a=(a_1,\ldots,a_M)^{T}$ symmetric as in \eqref{Sym_Pol_DS_Rewritten} or \eqref{Sym_Pol_DS_Rewritten_new}.

In conclusion let us mention that  we have  studied  dynamical systems describing identical particle, i.e. particles possessing the same characteristics (such as mass, charge or circulation). The polynomial method is also applicable to systems of distinct particles. In the latter case one should divide the particles into groups according to the values of mass, charge, circulation etc. and introduce polynomials for each group separately \cite{Demina25,Demina26, Demina27, Demina28, Demina29, Demina30}. Along with this it can be seen that polynomials with multiple roots satisfying partial differential equations give rise to dynamical systems describing behavior of distinct particles.

\section{Multi--particle dynamical systems corresponding to linear partial differential equations} \label{Nparticle_linear}

The polynomial method of solving a polynomial multi--particle dynamical system consists in finding polynomial solutions of the corresponding partial differential equation.
In many cases obtaining polynomial solutions of a linear partial differential equations is easier than those of nonlinear equations, especially at large values of the parameter $M$.

In this section we shall construct and solve a number of multi--particle dynamical systems that originate from linear  partial differential equations. Restricting ourselves with two--particle interactions (in the case $M>2$) we  consider second--order equations
\begin{equation}
\begin{gathered}
\label{Nparticle_PDE_linear}
\alpha_{0,2}(z,t)p_{tt}+\alpha_{0,1}(z,t)p_{t}+\alpha_{1,1}(z,t)p_{tz}+\alpha_{2,0}(z,t)p_{zz}+\alpha_{1,0}(z,t)p_{z}+\alpha_{0,0}(z,t)p=0,
\end{gathered}
\end{equation}
where the coefficient functions $\{\alpha(z,t)\}$ are polynomials in $z$. Suppose that  a polynomial $p(z,t)$ with simple roots is a solution of this equation. Substituting relations \eqref{Nparticle_relations2} into  equation \eqref{Nparticle_PDE_linear} gives  the following multi--particle dynamical system
\begin{equation}
\begin{gathered}
\label{Nparticle_PDE_linear_DS}
\alpha_{0,2}(a_k,t)a_{k,tt}+\left[\alpha_{0,1}(a_k,t)+2\alpha_{0,2}(a_k,t)\{\log c_0\}_t\right]a_{k,t}=\alpha_{1,0}(a_k,t)-\alpha_{1,1}(a_k,t)\{\log c_0\}_t\\
+\sum_{j=1,j\neq k}^{M}\frac{1}{a_k-a_j}
\left\{[2\alpha_{0,2}(a_k,t)a_{k,t}-\alpha_{1,1}(a_k,t)]a_{j,t}
+[2\alpha_{2,0}(a_k,t)-\alpha_{1,1}(a_k,t)a_{k,t}]\right\},
\end{gathered}
\end{equation}
where $k=1\ldots M$. Conversely, starting from system \eqref{Nparticle_PDE_linear_DS} we introduce a polynomial $p(z,t)$ of degree $M$ with roots at the particle positions (see \eqref{Polynomial_general_coeff}).  By means of relations \eqref{Nparticle_relations1} we obtain equalities $F[t,z,p_z,p_t,p_{tz},p_{tt}$, $p_{zz}]_{z=a_k}=0$, $k=1\ldots M$, where the polynomial $F$ is given by
\begin{equation}
\begin{gathered}
\label{Nparticle_PDE_linear_Calculs1}
F=\alpha_{0,2}(z,t)p_{tt}+\alpha_{0,1}(z,t)p_{t}+\alpha_{1,1}(z,t)p_{tz}+\alpha_{2,0}(z,t)p_{zz}+\alpha_{1,0}(z,t)p_{z}.
\end{gathered}
\end{equation}
Thus we see that the polynomial $p(z,t)$ satisfies the equation
\begin{equation}
\begin{gathered}
\label{Nparticle_PDE_linear_Calculs2}
\alpha_{0,2}(z,t)p_{tt}+\alpha_{0,1}(z,t)p_{t}+\alpha_{1,1}(z,t)p_{tz}+\alpha_{2,0}(z,t)p_{zz}+\alpha_{1,0}(z,t)p_{z}-P(z,t)p=0,
\end{gathered}
\end{equation}
where $P(z,t)$ is a polynomial in $z$ of degree: $\deg P=\deg F-M$. Note that in the linear case the polynomial $P$ does not depend on $p$ and its derivatives. In
order to identify $P$ as $\alpha_{0,0}(z,t)$ system \eqref{Nparticle_PDE_linear_DS} should be supplied by additional constrains, which we derive as described in the previous section. This procedure is equivalent to substituting equality \eqref{Polynomial_general_coeff} into equation \eqref{Nparticle_PDE_linear} and setting to zero the coefficients at $z^{M+\deg \alpha_{0,0}}$, $\ldots$, $z^M$.

A necessary and sufficient condition for a polynomial $p(z,t)$ to satisfy equation \eqref{Nparticle_PDE_linear} or \eqref{Nparticle_PDE_llinear}
is existence of truncated Laurent series in a neighborhood of the points $z=0$ and $z=\infty$. In fact, these series coincide and are given by \eqref{Polynomial_general_coeff}. Substituting expression \eqref{Polynomial_general_coeff} into equation \eqref{Nparticle_PDE_linear} and setting to zero coefficients at different powers of $z$ one obtains a linear system for the coefficients $c_{0}(t)$, $\ldots$, $c_{M}(t)$. As a rule these equations are differential. As soon as a polynomial $p(z,t)$ with simple roots that solves equation \eqref{Nparticle_PDE_linear} is found it is an algebraic problem to obtain solutions of the corresponding dynamical system.

Let us consider several examples. The following linear partial differential equation
\begin{equation}
\begin{gathered}
\label{Nparticle_PDE_linear_Class_Ort_Pol}
\beta_2p_{tt}+\beta_1p_t+\sigma(z)p_{zz}+\tau(z)p_z+\lambda p=0
\end{gathered}
\end{equation}
with $\beta_1$, $\beta_2$, $\lambda$  being constants and $\sigma(z)$, $\tau(z)$ being polynomials such that $\deg \sigma \leq2$, $\deg \tau\leq1$ possesses stationary polynomial solutions given by classical orthogonal polynomials (of course, under appropriate choices of the parameter $\lambda$). Equation \eqref{Nparticle_PDE_linear_Class_Ort_Pol} necessarily admits polynomial solutions if the coefficient $c_M(t)$ satisfies the equation
\begin{equation}
\begin{gathered}
\label{Nparticle_PDE_linear_Class_Ort_Pol_c_M}
\beta_2c_{0,tt}+\beta_1c_{0,t}+\left(\frac{M(M-1)}{2}\sigma_{zz}+M\tau_{z}+\lambda\right)c_0=0
\end{gathered}
\end{equation}
This equation helps to find the polynomial $P(z,t)$ in expression \eqref{Nparticle_PDE_linear_Calculs2} and to establish a correspondence between polynomial solutions of equation \eqref{Nparticle_PDE_linear_Class_Ort_Pol} given by \eqref{Polynomial_general_roots} and the following multi--particle dynamical system
\begin{equation}
\begin{gathered}
\label{Nparticle_DS_Class_Ort_Pol}
\beta_2a_{k,tt}+[\beta_1+2\beta_2\{\log c_0\}_t]a_{k,t}=\tau(a_k)+2\sum_{j=1,j\neq k}^{M}\frac{\beta_2a_{k,t}a_{j,t}+\sigma(a_k)}{a_k-a_j},\quad k=1\ldots M.
\end{gathered}
\end{equation}
By $\{p_m(z)\}$ we denote a sequence of classical orthogonal polynomials satisfying the equation
\begin{equation}
\begin{gathered}
\label{Nparticle_Class_Ort_Pol_EQN}
\sigma(z)p_{m,zz}+\tau(z)p_{m,z}+\lambda_m p_m=0
\end{gathered}
\end{equation}
The polynomials $\{p_m(z)\}$ are orthogonal with respect to the weight function
\begin{equation}
\begin{gathered}
\label{Eqn_Determinants_Orthogonal_Polynomials_rho}
\rho_{cop}(z)=\frac{1}{\sigma}\exp\left[\int \frac{\tau}{\sigma} dz\right].
\end{gathered}
\end{equation}
on the real interval $[a,b]$, which may be infinite or half--infinite, see table \ref{T:COP}.

\begin{table}[t]
    \caption{Classical orthogonal polynomials.} \label{T:COP}
       \begin{tabular}[pos]{|l|c|c|c|c|c|c|}
                \hline
                              $ $ &   $p_m(z)$  & $ \rho_{cop}(z)$ & $[a,b]$ & $\sigma(z)$ & $\tau(z)$ & $\lambda_m$\\
                                 \hline
               Hermite & $H_m(z)$ & $ \exp (-z^2)$ & $(-\infty,+\infty)$ & $1$ & $-2z$ & $2m$\\
                \hline
                 Laguerre & $L_m^{(\alpha)}(z)$, $\alpha>-1$ & $z^{\alpha}\exp (-z)$ & $[0,+\infty)$ & $z$ & $\alpha+1-z$ & $m$\\
                \hline
                 Jacobi &  $P_m^{(\alpha,\beta)}(z)$, $\alpha>-1$ & $(1-z)^{\alpha}(1+z)^{\beta}$ & $[-1,1]$ & $1-z^2$ & $\beta-\alpha-(\alpha$ & $m(m+\alpha$\\
                $ $ &  $\beta>-1$ & $ $ & $ $ & $ $ & $ +\beta+2)z$ & $ +\beta+1)$\\
                \hline
        \end{tabular}
\end{table}

In these designations polynomial in $z$ solutions of equation \eqref{Nparticle_PDE_linear_Class_Ort_Pol} can be presented in the form
\begin{equation}
\begin{gathered}
\label{Nparticle_PDE_linear_Class_Ort_Pol_Sols1}
p(z,t)=\sum_{m=0}^Mb_m(t)p_m(z),
\end{gathered}
\end{equation}
where the coefficients $b_0(t)$, $\ldots$, $b_M(t)$ satisfy the following linear ordinary differential equations
\begin{equation}
\begin{gathered}
\label{Nparticle_PDE_linear_Class_Ort_Pol_Sols2}
\beta_2b_{m,tt}(t)+\beta_1b_{m,t}(t)+(\lambda-\lambda_m)b_m=0,\quad m=0\ldots M
\end{gathered}
\end{equation}
Let us solve these equations. If $\beta_2=0$, then we obtain
\begin{equation}
\begin{gathered}
\label{Nparticle_PDE_linear_Class_Ort_Pol_Calc1}
b_m(t)=C_m\exp\left[\frac{(\lambda_m-\lambda)t}{\beta_1}\right],\quad m=0\ldots M,
\end{gathered}
\end{equation}
where $\{C_m\}$ are arbitrary constants. In the case $\beta_2\neq0$ finding solutions of the quadratic equations
\begin{equation}
\begin{gathered}
\label{Nparticle_PDE_linear_Class_Ort_Pol_Calc2}
\beta_2\kappa^2+\beta_1\kappa+(\lambda-\lambda_m)=0,\quad m=0\ldots M,
\end{gathered}
\end{equation}
we get
\begin{equation}
\begin{gathered}
\label{Nparticle_PDE_linear_Class_Ort_Pol_Calc3}
\kappa_m^{(1)}\neq \kappa_m^{(2)}\quad \Rightarrow \quad b_m(t)=C_m^{(1)}\exp[\kappa_m^{(1)}t]+C_m^{(2)}\exp[\kappa_m^{(1)}t],\quad m=0\ldots M\\
\kappa_m^{(1)}= \kappa_m^{(2)}\quad \Rightarrow \quad b_m(t)=\left[C_m^{(1)}+C_m^{(2)}t\right]\exp[\kappa_m^{(1)}t],\quad m=0\ldots M,\hfill
\end{gathered}
\end{equation}
where again $\{C_m^{(1)}\}$, $\{C_m^{(2)}\}$ are arbitrary constants.

Stationary equilibria of dynamical system \eqref{Nparticle_DS_Class_Ort_Pol} is described by algebraic relations
\begin{equation}
\begin{gathered}
\label{Nparticle_DS_Stationary_Class_Ort_Pol}
\tau(a_k)+2\sigma(a_k)\sum_{j=1,j\neq k}^{M}\frac{1}{a_k-a_j}=0,\quad k=1\ldots M.
\end{gathered}
\end{equation}
The unique solution of this system is given by the roots of the classical orthogonal polynomial~$p_M(z)$. Indeed, it follows from our results that the variables $a_1$, $\ldots$, $a_M$
are solutions of this system if and only if the monic polynomial $p(z)$ with roots at the points $a_1$, $\ldots$, $a_M$ satisfies equation \eqref{Nparticle_Class_Ort_Pol_EQN} with $m=M$. The unique (up to a constant multiplier, which does not affect the rools) polynomial solution of the latter equation is $p_M(z)$. Note that algebraic system \eqref{Nparticle_DS_Stationary_Class_Ort_Pol} being considered in the complex plane possesses only real solutions.

Further let us study some other interesting examples involving classical orthogonal polynomials. It was proved in article \cite{Demina27} that the Wronskians $P(z)=W[p_{i_1},\ldots,p_{i_k},p_{i_{k+1}}]$, $Q(z)=W[p_{i_1},\ldots,p_{i_k}]$, where $i_1,\ldots,i_{k+1}$, $k\in\mathbb{N^{+}}$ is a sequence of pairwise different nonnegative integer numbers,  satisfy the following equation
\begin{equation}
\begin{gathered}
\label{Eqn_Determinants_Orthogonal_Polynomials}\sigma\left\{P_{zz}Q-2P_xQ_z+PQ_{zz}\right\}+\left[\tau+\left(k-\frac12\right)\sigma_z\right]\left\{P_xQ-PQ_z\right\}
+ \frac{\sigma_z}{2}\{P_xQ+PQ_z\}\\
+\left[\frac{k(k-1)}{2}\sigma_{zz}+k\tau_z+\lambda_{i_{k+1}}\right]PQ=0.
\end{gathered}
\end{equation}
In the case $k=0$ we set $P(z)=p_{i_1}$, $Q(z)=1$. For further purposes we need the following theorem.
\begin{theorem}\label{Teorem:Nparticle_Orth_Zeros}
Neither the polynomials $P(z)=W[p_{i_1},\ldots,p_{i_k},p_{i_{k+1}}]$, nor the polynomial $Q(z)=W[p_{i_1},\ldots,p_{i_k}]$  have common roots with the polynomial $\sigma(z)$.
\end{theorem}

\begin{proof}

First of all let us prove that if the polynomial $Q(z)$ does not have common roots with the polynomial $\sigma(z)$ then neither does the polynomial $P(z)$. Suppose $z=z_0$ is a root of the polynomial $\sigma(z)$. Substituting the Tailor series in a neighborhood of the point $z=z_0$:
\begin{equation}
\begin{gathered}
\label{Eqn_Determinants_Orthogonal_Polynomials_Calc1}\sigma(z)=\sigma_1(z-z_0)+\sigma_2(z-z_0)^2,\quad \tau(z)=\tau_0+\tau_1(z-z_0),\\
Q(z)=\kappa_0+\sum_{m=1}^{\deg Q}\kappa_m(z-z_0)^m,\quad P(z)=\mu_r(z-z_0)^r+\sum_{m=r+1}^{\deg P}\mu_m(z-z_0)^m,\\
\end{gathered}
\end{equation}
where $\kappa_0\neq 0$, $\mu_r\neq0$, $r\in\mathbb{N^+}$, into equation \eqref{Eqn_Determinants_Orthogonal_Polynomials} and setting the lowest--order coefficient to zero yields the equality
\begin{equation}
\begin{gathered}
\label{Eqn_Determinants_Orthogonal_Polynomials_Calc2}r\kappa_0\mu_r\left\{(r+k-1)\sigma_1+\tau_0\right\}=0.
\end{gathered}
\end{equation}
In the case of the Jacobi entries of the Wronskians we get
\begin{equation}
\begin{gathered}
\label{Eqn_Determinants_Orthogonal_Polynomials_Jacobi_Calc2}z_0=1,\quad \sigma_1=-2,\quad\tau_0=-2(\alpha+1)\quad \Rightarrow \quad   r\kappa_0\mu_r\left\{r+k+\alpha\right\}=0,\\
z_0=-1,\quad \sigma_1=2,\quad\tau_0=2(\beta+1)\quad \Rightarrow \quad   r\kappa_0\mu_r\left\{r+k+\beta\right\}=0.\hfill
\end{gathered}
\end{equation}
From  the conditions $\alpha>-1$, $\beta>-1$ it follows that $r=0$. In the case of the  Laguerre entries of the Wronskians we have
\begin{equation}
\begin{gathered}
\label{Eqn_Determinants_Orthogonal_Polynomials_Laguerre_Calc2}z_0=0,\quad \sigma_1=1,\quad\tau_0=\alpha+1\quad \Rightarrow \quad   r\kappa_0\mu_r\left\{r+k+\alpha\right\}=0.
\end{gathered}
\end{equation}
The condition $\alpha>-1$ gives $r=0$. Further we use induction over $k$. For $k=1$, there is nothing to prove since the classical orthogonal polynomial $p_{i_1}$ does not have common roots with the polynomial $\sigma(z)$. This completes the proof.

\end{proof}

Fixing the sequence $i_1,\ldots,i_{k}$ let us consider the polynomials $P_l(z)=W[p_{i_1},\ldots,p_{i_k},p_{l}]$, $l\in\mathbb{N^{+}}$ $\setminus$ $I$,  $I=\{i_1,\ldots,i_k\}$. It turns out that such sequences of polynomials form orthogonal systems.

\begin{theorem}\label{Teorem:Nparticle_Orth}
The polynomials $P_l(z)$, $l\in\mathbb{N^{+}}$ $\setminus$ $\{i_1,\ldots,i_k\}$ are orthogonal with respect to the weight--function $\tilde{\rho}(z)=(\sigma^k\rho_{cop})/q^2$ on any simple directed  smooth or piecewise smooth curve $\Gamma$ with  parametrization $\gamma(s):$ $[s_a,s_b] \rightarrow \mathbb{\overline{C}}$, $[s_a,s_b]\subseteq\mathbb{R}$ provided that $\Gamma$ avoids the zeros of the polynomial $q(z)=W[p_{i_1},\ldots,p_{i_k}]$, the integrals
\begin{equation}
\begin{gathered}
\label{Eqn_Determinants_Orthogonality_Polynomials_New_Calc1}
\int_{\Gamma}z^n\tilde{\rho}(z)dz,\quad n\in \mathbb{N^{+}}
\end{gathered}
\end{equation}
are finite, and the endpoints $a=\gamma(s_a)$, $b=\gamma(s_b)$ are chosen in such a way that
\begin{equation}
\begin{gathered}
\label{Eqn_Determinants_Orthogonality_Polynomials_New_Calc2}
\lim_{s\rightarrow s_a+0}\tilde{\rho}(\gamma(s))\sigma(\gamma(s))\gamma^n(s)=0,\quad  \lim_{s\rightarrow s_b-0}\tilde{\rho}(\gamma(s))\sigma(\gamma(s))\gamma^n(s)=0,\quad n\in \mathbb{N^{+}}.
\end{gathered}
\end{equation}
\end{theorem}
\textit{Remark 1.} In theorem \ref{Teorem:Nparticle_Orth} orthogonality is understood in the non--Hermitian sense
\begin{equation}
\begin{gathered}
\label{Eqn_Determinants_Orthogonality_Polynomials_New}
\int_{\Gamma}P_l(z)P_n(z)\tilde{\rho}(z)dz=0,\quad n\neq l
\end{gathered}
\end{equation}
unless the curve $\Gamma$ is an interval $[a,b]$, possibly infinite or half--infinite, of the real line and $\tilde{\rho}(z)dz$ is a positive measure on $[a,b]$.

\textit{Remark 2.} If the parameters $\alpha$, $\beta$ are not integers, then one should introduce cetrain cuts in the complex plane in order to ensure possibilities of choosing single valued branches of the functions $z^{\alpha}$, $(1-z)^{\alpha}$, $(1+z)^{\beta}$.

\textit{Remark 3.} We do not exclude the case of closed curves, i.e. $\gamma(s_a)=\gamma(s_b)$. For example, one can use this theorem to find polynomials orthogonal on a circle in the complex plane.

\begin{proof}
We begin the proof by observing that substituting $P=q\psi$, $Q=q$, $i_{k+1}=l$ into expression \eqref{Eqn_Determinants_Orthogonal_Polynomials} gives the following  linear second order equation for the rational function~$\psi(z)$:
\begin{equation}
\begin{gathered}
\label{Eqn_Determinants_Orthogonal_Polynomials_New}
\sigma\psi_{l,zz}+\left(\tau+k\sigma_z\right)\psi_{l,z}+\left(u_k+\lambda_{l}\right) \psi_l=0,
\end{gathered}
\end{equation}
where we have introduced notation
\begin{equation}
\begin{gathered}
\label{Eqn_Determinants_Orthogonal_Polynomials_New_u}
u_k(z)=k\tau_z+\frac{k(k-1)}{2}\sigma_{zz}+\sigma_z\{\log q\}_z+2\sigma\{\log q\}_{zz}.
\end{gathered}
\end{equation}
Further using the standard technic designed for equations satisfied by the sequences of classical orthogonal polynomials we rewrite equation \eqref{Eqn_Determinants_Orthogonal_Polynomials_New} in the form
\begin{equation}
\begin{gathered}
\label{Eqn_Determinants_Orthogonal_Polynomials_New_selfad}
\frac{d}{d z}\left(\sigma\rho \psi_{l,z}\right)+\left(u_k+\lambda_{l}\right)\rho \psi_l=0,
\end{gathered}
\end{equation}
where the function $\rho(z)$ is defined by the relation
\begin{equation}
\begin{gathered}
\label{Eqn_Determinants_Orthogonal_Polynomials_New_rho_eqn}
\frac{d}{d z}\left(\sigma\rho \right)=\left(\tau+k\sigma_z\right)\rho.
\end{gathered}
\end{equation}
Dividing this equality by $\sigma\rho$ and integrating the result, we get $\rho(z)=\sigma^k\rho_{cop}$.

Multiplying equation by $\psi_n$, $n\neq l$ yields
\begin{equation}
\begin{gathered}
\label{Eqn_Determinants_Orthogonal_Polynomials_New_Calc1}
\left(\sigma\rho \psi_{l,\,z}\right)_z\psi_n+\left(u_k+\lambda_{l}\right)\rho \psi_l\psi_n=0,
\end{gathered}
\end{equation}
Reversing the subscripts, subtracting one equality from the other, and integrating the result gives
\begin{equation}
\begin{gathered}
\label{Eqn_Determinants_Orthogonal_Polynomials_New_Calc2}
(\lambda_n-\lambda_l)\int_{\Gamma} \psi_n(z)\psi_l(z) \rho(z)dz=-\left\{\rho\sigma \frac{W[P_n,P_l]}{q^2}\right\}_{a}^{b},
\end{gathered}
\end{equation}
The path of integration $\Gamma$ is chosen as given in the statement of the theorem. It follows from relations \eqref{Eqn_Determinants_Orthogonality_Polynomials_New_Calc2} that the right--hand side of expression \eqref{Eqn_Determinants_Orthogonal_Polynomials_New_Calc2} equals zero. Consequently, the sequence of rational functions $\psi_l(z)$, $l\in\mathbb{N^{+}}$ $\setminus$ $\{i_1,\ldots,i_k\}$  is an orthogonal system with weight $\rho(z)$ on the curve $\Gamma$. Recalling the definition of the function $\psi_l(z)$, we obtain that the  polynomials $P_l(z)$, $l\in\mathbb{N^{+}}$ $\setminus$ $\{i_1,\ldots,i_k\}$ are orthogonal with weight $\tilde{\rho}(z)=\rho(z)/q(z)^2$ on the curve $\Gamma$. This completes the proof.
\end{proof}

Polynomials, whose orthogonality  we have proved in theorem \ref{Teorem:Nparticle_Orth}, seem to belong to the class of exceptional orthogonal polynomials (for definitions and review of properties see \cite{Gomez01, Gomez02}). If the polynomial $q(z)$ does not have real zeros, then the curve $\Gamma$ in orthogonality condition \eqref{Eqn_Determinants_Orthogonality_Polynomials_New} can be chosen as the real interval $[a,b]$ and the endpoints may be taken as in the classical case.

The degrees of the polynomials $P_l(z)=W[p_{i_1},\ldots,p_{i_k},p_{l}]$, $l\in\mathbb{N^{+}}$ $\setminus$ $\{i_1,\ldots,i_k\}$, $q=W[p_{i_1},\ldots,p_{i_k}]$ can be calculated finding the highest powers of the Wronskians. The result is
\begin{equation}
\begin{gathered}
\label{Exc_Ort_Pols_Degrees}
\deg P_l=\sum_{m=1}^{k}i_m+l-\frac{k(k+1)}{2},\quad \deg q=\sum_{m=1}^{k}i_m-\frac{k(k-1)}{2}.
\end{gathered}
\end{equation}
Note that  the following condition:  $\deg P_l=l$ is valid not for all sequences of  orthogonal polynomials.

As an example let us take the Laguerre entries of the Wronskians. The polynomial $W[L_1^{(\alpha)},L_2^{(\alpha)}]=-\{z^2-2(\alpha+1)z+(\alpha+1)(\alpha+2)\}/2$ does not have real zeros provided that the inequality $\alpha>-1$ is satisfied. Consequently, setting $QL=W[L_1^{(\alpha)},L_2^{(\alpha)}]$ we obtain a family of polynomials  $PL_m=W[L_1^{(\alpha)},L_2^{(\alpha)},L_m^{(\alpha)}]$, $m\in\mathbb{N^+}\setminus\{1,2\}$, $\alpha>-1$ orthogonal with the weight $\tilde{\rho}(z)=(z^{\alpha+2}\exp[-z])/q^2$ on the real interval $[0,+\infty)$:
\begin{equation}
\begin{gathered}
\label{Eqn_Determinants_Orthogonal_Polynomials_Example_Laguerre}
\int_{0}^{\infty}PL_l(z)PL_n(z)\frac{z^{\alpha+2}e^{-z}}{QL^2}dz=0,\quad n\neq l.
\end{gathered}
\end{equation}
According to relation \eqref{Exc_Ort_Pols_Degrees} we have $\deg PL_m=m$. Let us mention an interesting property of these polynomials. We observe that the polynomial $PL_m(z)$ possesses $m-\deg QL$ simple zeros on the interval of orthogonality. For several plots see figure \ref{F:LG_exceptional_OP}. Roots that are located on the curve of orthogonality are called regular.


\begin{figure}[t]
 \centerline{
 \subfigure[$q(z)$]{\epsfig{file=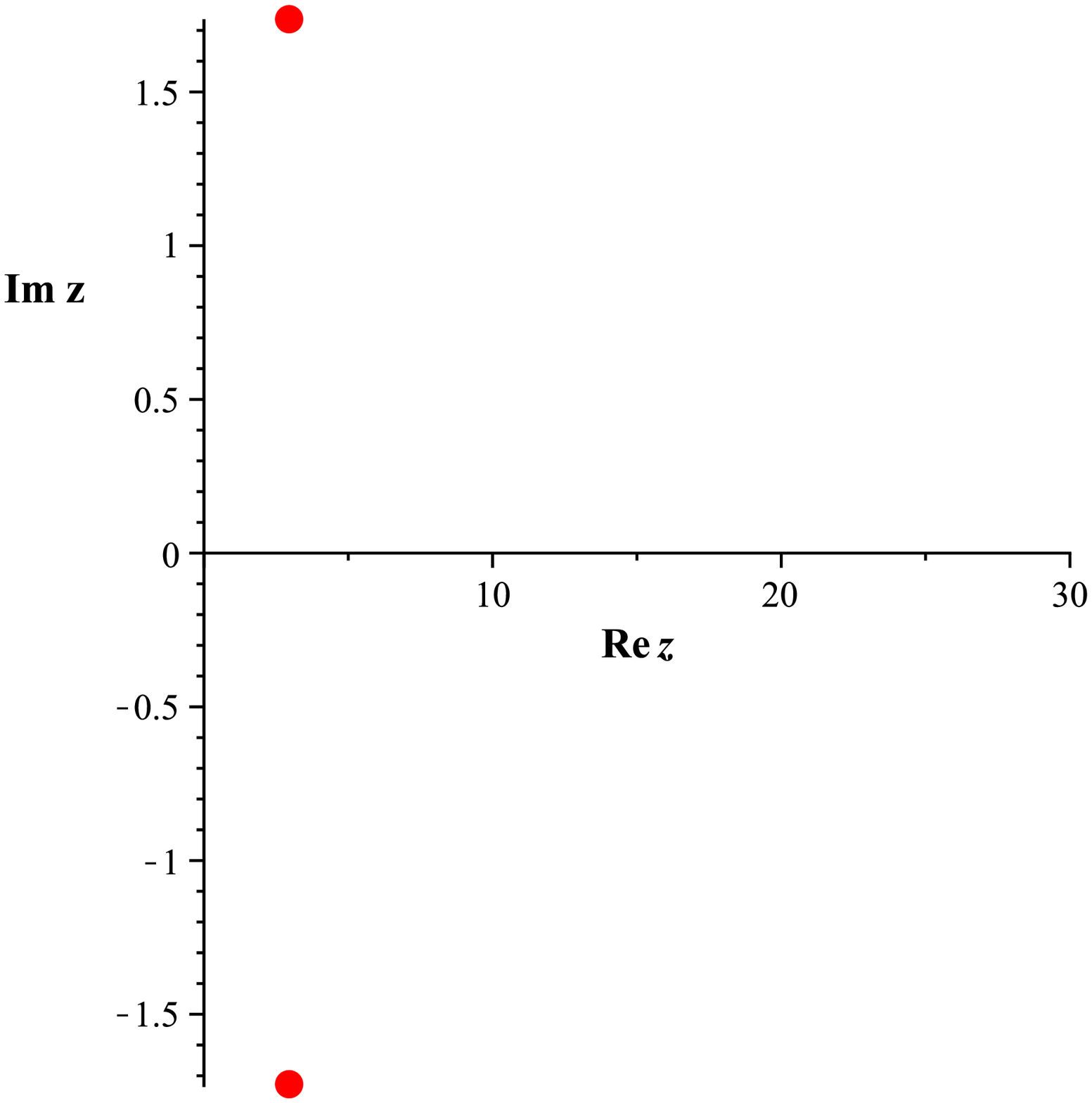,width=60mm}\label{}}
 \subfigure[$P_7(z)$]{\epsfig{file=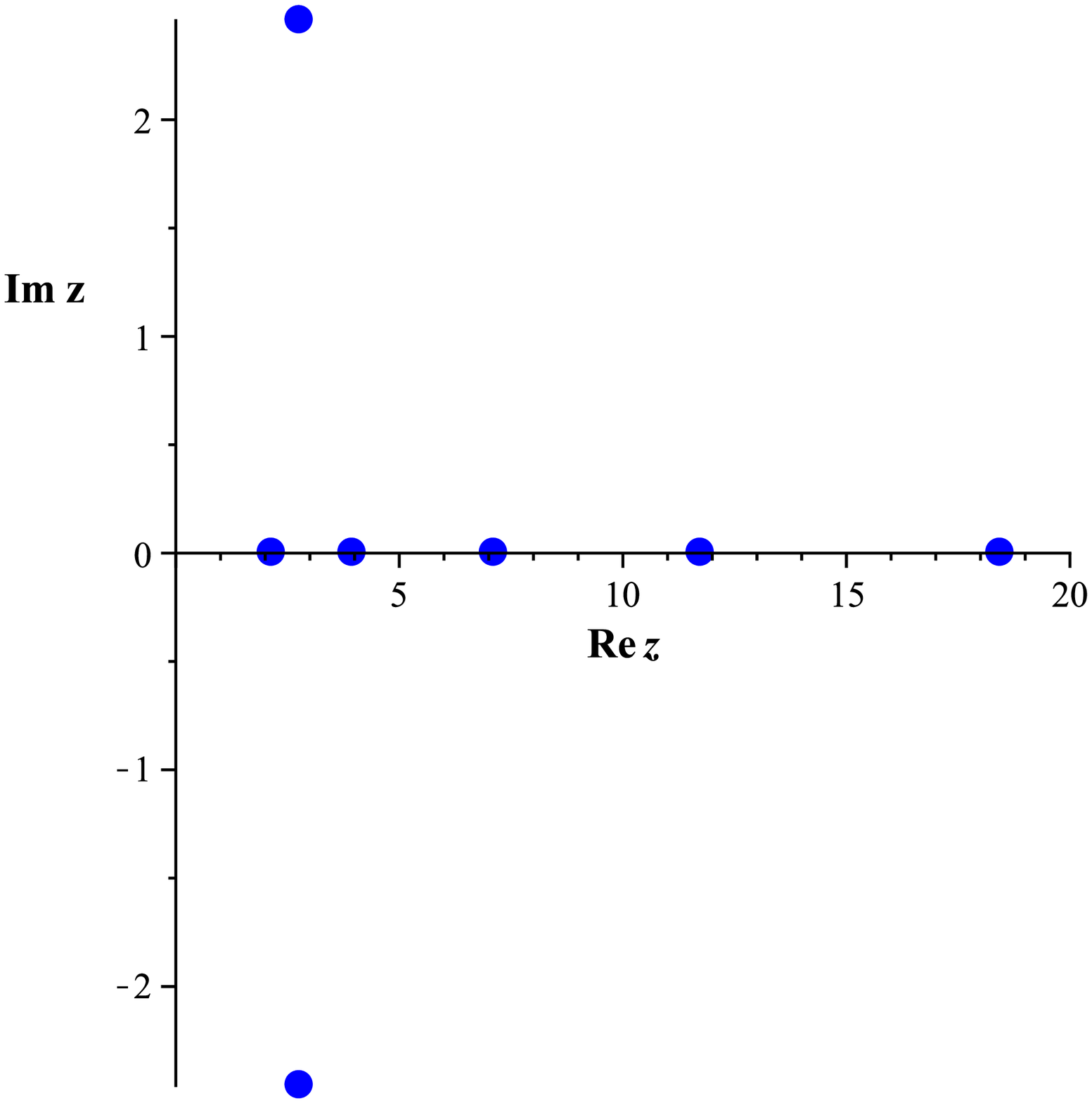,width=60mm}\label{}}
 \subfigure[$P_{12}(z)$]{\epsfig{file=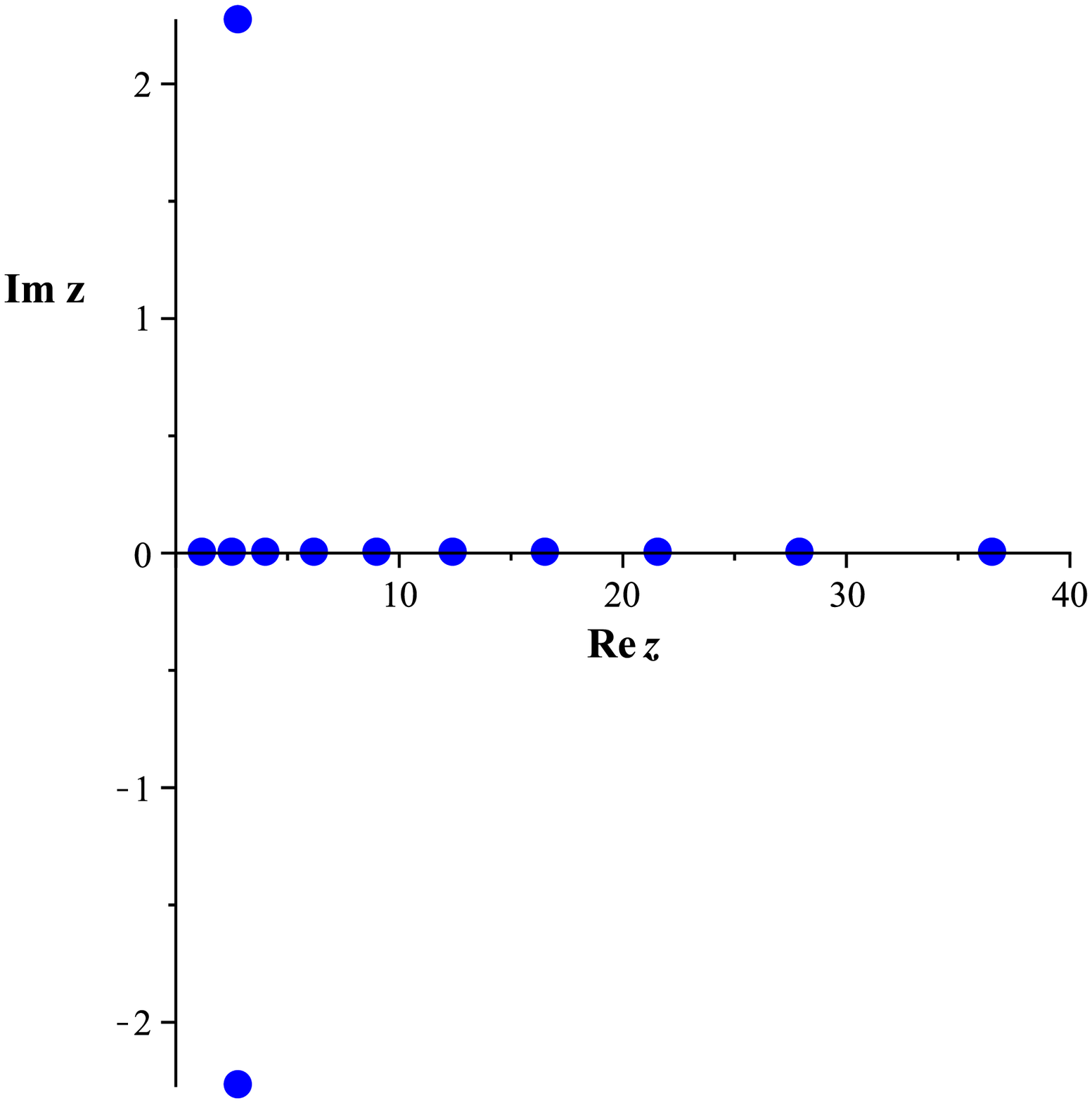,width=60mm}\label{}}}
    \caption{Roots of the polynomials $q(z)$, $\{P_m(z)\}$, $\alpha=2$.}
 \label{F:LG_exceptional_OP}
\end{figure}

If the degrees of the polynomials in coefficients of equation \eqref{Nparticle_PDE_linear} are the same as in expression \eqref{Nparticle_PDE_linear_Class_Ort_Pol}, i.e. $\deg\alpha_{0,2}=0$, $\deg\alpha_{0,1}=0$, $\deg\alpha_{0,0}=0$, $\deg\alpha_{1,1}\leq1$, $\deg\alpha_{1,0}\leq1$, $\deg\alpha_{2,0}\leq2$, then substituting equality \eqref{Polynomial_general_coeff} into equation \eqref{Nparticle_PDE_linear} yields $M+1$ linear second--order ordinary differential equations on $M+1$ coefficients $c_0(t)$, $\ldots$, $c_M(t)$. In principle, these equations can be solved step by step (exact resolution, of course, depends on the functions of $t$ involved in the coefficients $\{\alpha(z,t)\}$). However, if the polynomials in coefficients of equation \eqref{Nparticle_PDE_linear} are of higher degrees, then one obtains more than $M+1$ equations on $M+1$ coefficients $c_0(t)$, $\ldots$, $c_M(t)$. The problem to construct polynomial solutions of corresponding partial differential  equations (or to prove their existence) becomes more complicated.  As a result, it is not an easy task to find solutions of related dynamical systems (at least, by means of the polynomial method). We can use  orthogonal polynomials $P_l(z)$, $l\in\mathbb{N^{+}}$ $\setminus$ $\{i_1,\ldots,i_k\}$ to solve this problem whenever  the $z$--part of a partial differential equation coincides with the equation for the system of orthogonal polynomials in question and $\alpha_{0,2}=\beta_2 q$, $\alpha_{0,1}=\beta_1 q$.

Polynomials with simple roots that solve
the following linear partial differential equation
\begin{equation}
\begin{gathered}
\label{PDE_exeptional_OP}
q\{\beta_2p_{tt}+\beta_1p_t+\sigma p_{zz}\}+[(\tau+k\sigma_z)q-2\sigma q_z]p_z+\left[\sigma q_{zz}-(\tau+(k-1)\sigma_z)q_z \phantom{\frac{1}{2}}\right.\\
\left.+\left\{\frac{k(k-1)}{2}\sigma_{zz}+k\tau_z+\lambda\right\}q\right]p=0
\end{gathered}
\end{equation}
 give rise to the dynamical system of the form
\begin{equation}
\begin{gathered}
\label{Nparticle_DS_exeptional_OP}
q(a_k)\left\{\beta_2a_{k,tt}+[\beta_1+2\beta_2\{\log c_0\}_t]a_{k,t}\right\}=2q(a_k)\sum_{j=1,j\neq k}^{M}\frac{\beta_2a_{k,t}a_{j,t}+\sigma(a_k)}{a_k-a_j}\\
+\left\{\tau(a_k)+k\sigma_z(a_k)\right\}q(a_k)
-2\sigma(a_k)q_z(a_k)
,\quad k=1\ldots M.
\end{gathered}
\end{equation}
Suppose we wish to establish a correspondence between polynomials with simple roots that satisfy equation \eqref{PDE_exeptional_OP} and a dynamical system. It means that we should identify the polynomial $P(z,t)$ in \eqref{Nparticle_PDE_linear_Calculs2} as it is given in equation \eqref{PDE_exeptional_OP}. For this aim we substitute expression \eqref{Polynomial_general_coeff} into equation \eqref{PDE_exeptional_OP} and set to zero the coefficients at different powers of $z$. We need to take equations at $z^{M+\deg q}$, $\ldots$, $z^{M}$. This gives additional constrains on dynamical variables $a_1(t)$, $\ldots$, $a_M(t)$. For example, in the case $\sigma=z$, $\tau=\alpha+1-z$, $q=W[L_1^{(\alpha)},L_2^{(\alpha)}]$, $k=2$ we obtain
\begin{equation}
\begin{gathered}
\label{Nparticle_Exceptional_Laguerre}
\beta_2c_{0,tt}+\beta_1c_{0,t}+(\lambda-M)c_0=0,\hfill \\
\beta_2c_{1,tt}+\beta_1c_{1,t}+(\lambda+1-M)c_1+\delta_2c_0=0,\hfill \\
\beta_2c_{2,tt}+\beta_1c_{2,t}+(\lambda+2-M)c_2+\delta_3c_1-4(\alpha+1)Mc_0=0
\end{gathered}
\end{equation}
where $\delta_l=(M-l)^2+2(M-l)+M\alpha$. Further, we note that $c_1=-\{a_1+\ldots +a_M\}c_0$, $c_2=\{a_1a_2+a_1a_3+\ldots +a_{M-1}a_M\}c_0$, see also relations \eqref{Sym_Pol_Sym_Func_Coeff_non_monic} in section \ref{Symmetric_DS}.

Equation \eqref{PDE_exeptional_OP} possesses polynomial solutions of the form
\begin{equation}
\begin{gathered}
\label{Nparticle_DS_exeptional_OP_ES}
p(z,t)=\sum_{m=0,\,m\not \in I}^{M}b_m(t)P_m(z)
\end{gathered}
\end{equation}
with coefficients $b_m(t)$, $m\in\mathbb{N^+}\setminus\{I\}$ given by \eqref{Nparticle_PDE_linear_Class_Ort_Pol_Calc1}, \eqref{Nparticle_PDE_linear_Class_Ort_Pol_Calc3}. Thus, we see that if such a polynomial does not have multiple roots, then these roots satisfy dynamical equations~\eqref{Nparticle_DS_exeptional_OP}.

We would like to remark that derived systems of orthogonal polynomials miss polynomials of certain degrees. Consequently, as soon as polynomial solution \eqref{Nparticle_DS_exeptional_OP_ES} is constructed one comes across a question whether expression \eqref{Nparticle_DS_exeptional_OP_ES} is the most general polynomial solution of partial differential equation \eqref{PDE_exeptional_OP}. For certain systems of  orthogonal polynomials including the family $PL_m=W[L_1^{(\alpha)},L_2^{(\alpha)},L_m^{(\alpha)}]$, $m\in\mathbb{N^+}\setminus\{1,2\}$ we have checked that it is indeed the case. For this aim one can substitute the general expression for polynomial solutions (in terms of orthogonal polynomials under consideration and  powers of $z$) into equation \eqref{PDE_exeptional_OP} and set to zero the corresponding coefficients. For example, in the case of the polynomials $PL_m=W[L_1^{(\alpha)},L_2^{(\alpha)},L_m^{(\alpha)}]$, $m\in\mathbb{N^+}\setminus\{1,2\}$ we have substituted the following representation
\begin{equation}
\begin{gathered}
\label{Nparticle_DS_exeptional_OP_ES_alt}
p(z,t)=\sum_{m=0,\,m\neq 1,2}^{M}b_m(t)PL_m(z)+d_2(t)z^2+d_1(t)z+d_0(t),
\end{gathered}
\end{equation}
where for $d_0(t)$ we can take a partial solution of the corresponding ordinary differential equation. Further, we have shown that $d_2(t)=0$, $d_1(t)=0$, $d_0(t)=0$.

\section{Two and three--particle dynamical systems} \label{Sec:MP23}

In this section we consider the problem of integrating a class of two and three--particle polynomial dynamical systems. We begin with the  two--particle  non--autonomous dynamical systems given by
\begin{equation}
\begin{gathered}
\label{P2_DS}
a_{1,t}=\beta_1(t)a_1^2+\beta_2(t)a_2^2+\beta_3(t)a_1a_2+\beta_4(t)a_1+\beta_5(t)a_2+\beta_6(t),\\
a_{2,t}=\beta_1(t)a_2^2+\beta_2(t) a_1^2+\beta_3(t) a_1a_2+\beta_4(t)a_2+\beta_5(t)a_1+\beta_6(t).
\end{gathered}
\end{equation}
Let us suppose that the functions $\beta_1(t)$, $\ldots$, $\beta_6(t)$ are entire or meromorphic.  In order to apply the polynomial method we need the following relation: $a_j=\{z-p_z\}_{x=a_k}$, $j\neq k$. Applying the polynomial method we obtain the following  partial differential equation
\begin{equation}
\begin{gathered}
\label{P2_PDE}
p_t+\beta_2p_z^3-\{(2\beta_2+\beta_3)z+\beta_5\}p_z^2+\{(\beta_1+\beta_2+\beta_3)z^2+(\beta_4+\beta_5)z+\beta_6(t)\}p_z\\
-\{f_0(t)z+f_1(t)\}p=0.
\end{gathered}
\end{equation}
For subsequent convenience we introduce the new function: $\gamma(t)=\beta_1(t)+\beta_2(t)-\beta_3(t)$. Substituting \eqref{Polynomial_general_coeff} with $c_0(t)=1$, $M=2$ into this equation yields the equalities
\begin{equation}
\begin{gathered}
\label{P2_Calc1}
f_0=2\gamma,\quad f_1=(3\beta_2-\beta_1-\beta_3)w+2(\beta_4-\beta_5),\quad c_1=w,\\
 c_2=\frac{1}{2\gamma}\{w_t+(\beta_2+\beta_1)w^2-(\beta_4+\beta_5)w+2\beta_6(t)\}
\end{gathered}
\end{equation}
valid in the case $\gamma\not\equiv0$. Note that we do not exclude the situations, when the function $\gamma(t)$ possesses isolated zeros. The function $w(t)$ satisfies the following ordinary differential equation
\begin{equation}
\begin{gathered}
\label{P2_Calc2}
w_{tt}+(4\beta_1-\gamma)ww_t+\{\beta_1-\beta_2\}\{2(\beta_1+\beta_2)-\gamma\}w^3+\{\beta_5-3\beta_4-h\}w_t\\
+\{(\beta_4-\beta_5)\gamma
+4(\beta_2\beta_5
-\beta_1\beta_4)
+(\beta_1+\beta_2)_t-(\beta_1+\beta_2)h\} w^2
 +\{2(\beta_4^2-\beta_5^2)\\
 +4(\beta_1-\beta_2)\beta_6
 -(\beta_4+\beta_5)_t+(\beta_4+\beta_5)h\}w
 +2\{\beta_{0,t}-h\}\beta_6\}+4(\beta_5-\beta_4)=0,
\end{gathered}
\end{equation}
where $h(t)=(\log\gamma)_t$. This equation belongs to the class of second--order differential equations studied by P. Painlev\'{e}, B. Gambier, and their colleagues \cite{Painleve01, Gambier01}, see the introduction.

Consequently, if the functions $\beta_1(t)$, $\ldots$, $\beta_6(t)$   are taken in such a way that equation \eqref{P2_Calc2} becomes one of the canonical equations found by P. Painlev\'{e} and B. Gambier or is reducible to a canonical equation via the M\"{o}bius transformation
\begin{equation}
\begin{gathered}
\label{P2_equivalence_transformation}
w(t)=\frac{\lambda_1(t)W(\tilde{t})+\lambda_2(t)W(\tilde{t})}{\lambda_3(t)W(\tilde{t})+\lambda_4(t)W(\tilde{t})},\quad \tilde{t}=\Psi(t),\quad
\lambda_1\lambda_4-\lambda_2\lambda_3\neq0,
\end{gathered}
\end{equation}
where $\Psi(t)$, $\lambda_j(t)$, $j=1$ $\ldots$ $4$ are locally analytic functions, then we conclude that polynomial solution \eqref{Polynomial_general_coeff} with $c_0(t)=1$, $M=2$ of equation \eqref{P2_Calc2} is constructed. The original dynamical system can be integrated finding the roots of the quadratic  equation.

For example, equation \eqref{P2_Calc2}  coincides with the second Painlev\'{e} equation
\begin{equation}
\begin{gathered}
\label{P2_exact}
w_{tt}-2w^3-tw\pm\frac12=0
\end{gathered}
\end{equation}
provided that the functions $\beta_1(t)$, $\ldots$, $\beta_5(t)$ are constants and
\begin{equation}
\begin{gathered}
\label{P2_Calc3}
\beta_2=\beta_1\pm 1,\quad \beta_3=-2\beta_1\pm 1,\quad \beta_4=0,\quad \beta_5=0,\quad \beta_6(t)=\pm\frac{t}{4}.
\end{gathered}
\end{equation}
Note that the second Painlev\'{e} equation in the form \eqref{P2_exact} possesses one--parametric family of exact solutions expressible via the Airy functions.
Some exact solutions of equation \eqref{P2_Calc2} in autonomous case are given in article \cite{ML01}.

If $\gamma\equiv0$, then the functions $f_0$, $f_1$ take the form
\begin{equation}
\begin{gathered}
\label{P2_Calc1_case2}
f_0=0,\quad f_1=2(\beta_2-\beta_1)w+2(\beta_4-\beta_5).
\end{gathered}
\end{equation}
The coefficient $c_1=w$ satisfies the  first--order ordinary differential equation
\begin{equation}
\begin{gathered}
\label{P2_Calc2_case2}
w_t+(\beta_1+\beta_2)w^2-(\beta_4+\beta_5)w+2\beta_6=0,
\end{gathered}
\end{equation}
which is linear whenever $\beta_1\equiv-\beta_2$ and a Riccati equation in the case $\beta_1\not\equiv-\beta_2$. The coefficient $c_2=y$ solves the following linear nonhomogeneous differential equation
\begin{equation}
\begin{gathered}
\label{P2_Calc3_case2}
y_t+2\{(\beta_1-\beta_2)w+\beta_5-\beta_4\}y+\beta_2w^3-\beta_5w^2+\beta_6w=0
\end{gathered}
\end{equation}
Note that Riccati equation \eqref{P2_Calc2_case2} is linearizable via the substitution
\begin{equation}
\begin{gathered}
\label{P2_Calc4_case2}
w=\frac{\psi_t}{(\beta_1+\beta_2)\psi}, \quad \beta_1\neq-\beta_2
\end{gathered}
\end{equation}
with the function $\psi(t)$ satisfying the equation
\begin{equation}
\begin{gathered}
\label{P2_Calc5_case2}
\psi_{tt}-\{\beta_4+\beta_5+(\log[\beta_1+\beta_2])_t\}\psi_t+2\{\beta_1+\beta_2\}\beta_6\psi=0.
\end{gathered}
\end{equation}

Further let us turn to a three--dimensional case. The following polynomial multi--particle dynamical system
\begin{equation}
\begin{gathered}
\label{General_DS_DH}
a_{k,t}=\beta_1\prod_{j=1,j\neq k}^{3}a_j+\beta_2a_k\sum_{j=1,j\neq k}^3a_j+\beta_3\sum_{j=1,j\neq k}^3a_j^2+\beta_4a_k^2\\
+\beta_5\sum_{j=1,j\neq k}^3a_j+\beta_6
a_k+\beta_7,\quad k=1,2,3
\end{gathered}
\end{equation}
contains as partial cases several systems of great practical importance. If $\beta_1=1$, $\beta_2=-1$, and $\beta_j=0$, $j=2$, $\ldots$, $7$ system \eqref{General_DS_DH} is exactly the classical Darboux--Halphen system
\begin{equation}
\begin{gathered}
\label{Euler_DH}
a_{k,t}=\prod_{j=1,j\neq k}^{3}a_j-a_k\sum_{j=1,j\neq k}^3a_j,\quad k=1,2,3,
\end{gathered}
\end{equation}
which first appeared in Darboux's works on triply orthogonal surfaces \cite{Darboux01} and was later solved by Halphen \cite{Halphen01, Halphen02}.
In successive studies, the classical Darboux--Halphen system  has arisen as the Einstein field
equations for a diagonal self--dual Bianchi--IX metric with Euclidean signature and in the similarity reductions of associativity equations on a
three--dimensional Frobenius manifold \cite{Gibbons01}.

The Euler's dynamical system
\begin{equation}
\begin{gathered}
\label{Euler}
\xi_{1,t}=\sigma_1\xi_2\xi_3+\beta_6\xi_1+\sqrt{\sigma_1}\beta_7,\\
\xi_{2,t}=\sigma_2\xi_1\xi_3+\beta_6\xi_2+\sqrt{\sigma_2}\beta_7,\\
\xi_{3,t}=\sigma_3\xi_1\xi_2+\beta_6\xi_3+\sqrt{\sigma_3}\beta_7,
\end{gathered}
\end{equation}
describing the rotation of a rigid body can be transformed into the multi--particle form
\begin{equation}
\begin{gathered}
\label{Euler_MP}
a_{k,t}=\beta_1\prod_{j=1,j\neq k}^{3}a_j+\beta_6a_k+\beta_7,\quad k=1,2,3
\end{gathered}
\end{equation}
by means of the transformation $\xi=Ba$ with $B=\text{diag}(\sqrt{\sigma_1},\sqrt{\sigma_2},\sqrt{\sigma_3})$. In equations \eqref{Euler_MP} we use the designation $\beta_1=\det B$.
Dynamical system \eqref{Euler_MP} is also of the form \eqref{General_DS_DH}.

In order to apply the polynomial method to system \eqref{Euler_MP} we need the following equalities
\begin{equation}
\begin{gathered}
\label{DS_relations_three_dim}
\sum_{j=1,j\neq k}^3a_j=\left\{2z-\frac12 p_{zz}\right\}_{z=a_k},\quad \prod_{j=1,j\neq k}^3a_j=\left\{p_z-\frac12z p_{zz}+z^2\right\}_{z=a_k},\\
\sum_{j=1,j\neq k}^3a_j^2=\left\{\frac14 p_{zz}^2-zp_{zz}-2p_z+2z^2\right\}_{z=a_k}.
\end{gathered}
\end{equation}
The polynomial method gives the  partial differential equation of the form
\begin{equation}
\begin{gathered}
\label{Euler_PDE}
p_t+\frac{\beta_3}{4}p_zp_{zz}^2-\frac12\left\{(2\beta_3+\beta_2+\beta_1)z+\beta_5\right\}p_zp_{zz}+\{\beta_1-2\beta_3\}p_z^2
+\{(\beta_1+\beta_4\\
+2\beta_2+2\beta_3)z^2+(2\beta_5+\beta_6)z+\beta_7\}p_z-\{f_0(t)z+f_1(t)\}p=0.
\end{gathered}
\end{equation}
We are interested in third--degree polynomials with simple roots that solve this equation. Substituting \eqref{Polynomial_general_coeff} with $c_0(t)=1$, $M=3$ into equation \eqref{Euler_PDE}, we get relations
\begin{equation}
\begin{gathered}
\label{Euler_Calc1}
f_0=3(\beta_1+\beta_4-\beta_2-\beta_3),\quad f_1=(\beta_3-\beta_4+2\beta_1-2\beta_2)c_1+3(\beta_6-\beta_5)
\end{gathered}
\end{equation}
and three equations for the coefficient functions $c_1(t)$, $c_2(t)$, $c_3(t)$:
\begin{equation}
\begin{gathered}
\label{gen_DH_Calc2}
c_{1,t}+(2\beta_3+\beta_4)c_1^2-(2\beta_5+\beta_6)c_1+\gamma_1c_2+3\beta_7=0,\\
c_{2,t}+2\beta_3c_1^3-2\beta_5c_1^2+(\beta_1+\beta_2+\beta_4-5\beta_3)c_1c_2+2\beta_7c_1+2(\beta_5-\beta_6)c_2+3\gamma_2c_3=0,\\
c_{3,t}+\beta_3c_1^2c_2-\beta_5c_1c_2+(2\beta_2+\beta_4-2\beta_1-\beta_3)c_1c_3+(\beta_1-2\beta_3)c_2^2+\beta_7c_2+3(\beta_5-\beta_6)c_3=0.
\end{gathered}
\end{equation}
Here we have introduced notation $\gamma_1=\beta_1+2\beta_2-2\beta_4-4\beta_3$, $\gamma_2=\beta_3+\beta_2-\beta_1-\beta_4$. If $\gamma_1\not\equiv 0$ and $\gamma_2\not\equiv0$, then we  solve the first equation with respect to $c_2$, the second equation with respect to $c_3$ and substitute the resulting relations into the third equation. This yields the following third--order ordinary differential equation
\begin{equation}
\begin{gathered}
\label{gen_DH_ODE}
\gamma_1w_{ttt}+(\delta_1 w+\delta_5)w_{tt}+\delta_2 w_t^2+(\delta_3w^2+\delta_6w+\delta_7)w_{t}+\delta_4w^4\\
 +\delta_8w^3+\delta_9w^2+\delta_{10}w+\delta_{11}=0.
\end{gathered}
\end{equation}
for the function $c_1=w$. The coefficients of equation \eqref{gen_DH_ODE} are given in table \ref{table:DS_DH_coeff}. As soon  as a solution $w(t)$ of this equation is found, the coefficients $c_2$, $c_3$ are polynomially expressible via $w(t)$ and its derivatives. Again we do not exclude the situations, when the functions $\gamma_1(t)$, $\gamma_2(t)$  possess isolated zeros.

Let us consider the case $\gamma_1\equiv0$. If $\beta_4\equiv-2\beta_3$, then the first equation in \eqref{gen_DH_Calc2} is linear, otherwise this equation is the Riccati equation linearizable via the substitution
\begin{equation}
\begin{gathered}
\label{gen_DH_Calc3}
c_1=\frac{\psi_t}{(2\beta_3+\beta_4)\psi}, \quad \beta_4\neq-2\beta_3
\end{gathered}
\end{equation}
with the function $\psi(t)$ satisfying the equation
\begin{equation}
\begin{gathered}
\label{P2_Calc5_case2}
\psi_{tt}-(2\beta_5+\beta_6)\psi_t+3(2\beta_3+\beta_4)\beta_7\psi=0.
\end{gathered}
\end{equation}
Further, if $\gamma_2\not\equiv0$, then  solving the second equation in \eqref{gen_DH_Calc2} with respect to $c_3$ and substituting the result into the third equation, we obtain a second--order non-autonomous ordinary differential equation for the function $c_2$. If in addition to $\gamma_1$ the coefficient $\gamma_2$ identically equals zero, then the second equation in \eqref{gen_DH_Calc2} is a first--order linear inhomogeneous equation. Solving this equation we are left with the third equation in \eqref{gen_DH_Calc2}, which also becomes a first--order linear inhomogeneous equation.

Ordinary differential equation \eqref{gen_DH_ODE} belongs to the class of third--degree equations studied by J. Chazy \cite{Chazy01}, see the introduction.
In article \cite{Cosgrove01} C.M. Cosgrove developed the works of J. Chazy and collected all third--order ordinary differential equations in polynomial class having the Painlev\'{e} property. The general solutions of these equations are known. If the parameters of dynamical system \eqref{General_DS_DH}  are taken in such a way that equation \eqref{gen_DH_ODE} is exactly one from those given in \cite{Cosgrove01} or equation \eqref{gen_DH_ODE} is equivalent to one of the equations from article \cite{Cosgrove01}, then one can find the coefficient $c_1(t)$ and, consequently, coefficients $c_2(t)$, $c_3(t)$. Thus we come to a conclusion that polynomial solution \eqref{Polynomial_general_coeff} with $c_0(t)=1$, $M=3$ of equation \eqref{Euler_PDE} is constructed. The original dynamical system can be integrated finding the roots of this polynomial. Recall that two equations in polynomial class are regarded as equivalent if their solutions $w$, $\tilde{w}$ are related via the transformation
\begin{equation}
\begin{gathered}
\label{DS_equivalence_transformation}
w(t)=\lambda_0(t)W(\tilde{t})+\lambda_1(t),\quad \tilde{t}=\Psi(t).
\end{gathered}
\end{equation}
In principle one may consider M\"{o}bius transformation \eqref{P2_equivalence_transformation},  but with possible loss of polynomial dependance of the function $R$ in \eqref{Chazy_class}.

It is an interesting fact that C.M. Cosgrove has also given a number of third--order equations of the form \eqref{gen_DH_ODE} that are not of Painlev\'{e}--type but nevertheless that are integrable in the sense that one can find their general solutions \cite{Cosgrove01}.

Let us examine some interesting partial cases. In what follows we shall consider the autonomous case, i.e. we suppose that all the functions $\beta_j(t)$, $j=1$, $\ldots$, $7$ are constants.

Setting $\beta_2=\beta_1$, $\beta_3=-\beta_1$, $\beta_4=-\beta_1$, $\beta_1\neq0$ in \eqref{gen_DH_ODE} we obtain the  linear third--order equation
\begin{equation}
\begin{gathered}
\label{gen_DH_ODE_partial case_linear}
w_{ttt}+3(\beta_5-2\beta_6)w_{tt}+(4\beta_5+11\beta_6)(\beta_6-\beta_5)w_t-6\{(2\beta_5+\beta_6)w-3\beta_7\}(\beta_5-\beta_6)^2=0.
\end{gathered}
\end{equation}

If all the parameters in \eqref{General_DS_DH} are chosen in such a way that system \eqref{General_DS_DH} becomes the classical Darboux--Halphen system \eqref{Euler_MP}, then the function $y=2w$ satisfies the Chazy-III equation
\begin{equation}
\begin{gathered}
\label{DH_Chazy3}
y_{ttt}-2yy_{tt}+3y_t^2=0.
\end{gathered}
\end{equation}
This equation is a famous Painlev\'{e}--type equation admitting a movable natural barrier. Its general solution can be obtained inverting a hypergeometric function or using Schwarzian
triangle functions \cite{Cosgrove01}. Setting $\beta_3=-(\beta_1+\beta_2)/2$, $\beta_5=0$, $\beta_6=0$, $\beta_7=0$ in equation \eqref{gen_DH_ODE},
we see that the function $y=-2(\beta_4+\beta_2)w$, $\beta_2\neq-\beta_4$ satisfies the following third--order equation
\begin{equation}
\begin{gathered}
\label{DH_Chazy12}
y_{ttt}-2yy_{tt}+3y_t^2+\sigma(6y_t-y^2)^2=0,\quad \sigma=\frac{(2\beta_4-\beta_1)^2-\beta_2(7\beta_2+4\beta_4+6\beta_1)}{16(\beta_2+\beta_4)(3\beta_1-2\beta_4+4\beta_2)},
\end{gathered}
\end{equation}
which is the Chazy-XII equation if $\sigma=4/(n^2-36)$, $n\in \mathbb{N}/\{1,6\}$. System \eqref{General_DS_DH} with $\beta_3=-(\beta_1+\beta_2)/2$, $\beta_5=0$, $\beta_6=0$, $\beta_7=0$ can be regarded as a generalization of the classical Darboux--Halphen system. For details on the Chazy-XII equation see \cite{Cosgrove01}.

\begin{table}[t]
       \begin{tabular}[pos]{|l|}
        \hline
        $\delta_1=\frac{\gamma_1}{3}(2\gamma_1+5\gamma_2+21\beta_4-3\beta_3)$\\
         $\delta_2=\frac23(\gamma_1^2+\gamma_1\gamma_2-3\gamma_2^2)+2\gamma_1(\beta_3+2\beta_4)$\\
        $\delta_3=\frac23\gamma_1\gamma_2(2\gamma_1-\gamma_2)+18\gamma_1(\beta_4^2-\beta_3^2)+8\gamma_2(\beta_4\gamma_1-\beta_3\gamma_2)
+4\beta_4(\gamma_1^2-\gamma_2^2)+4\gamma_1\beta_3(\gamma_2-\gamma_1)$\\
        $\delta_4=\frac23(\gamma_1+6\beta_3+3\beta_4)(3\gamma_1\{\beta_3-\beta_4\}^2+2\gamma_1\gamma_2\{\beta_3-\beta_4\}-2\beta_3\gamma_2^2)$\\

$\delta_5=3\gamma_1(\beta_5-2\beta_6)-2\gamma_{1,t}-\gamma_1h_2$\\
$\delta_6=6\gamma_1\{\beta_3+\beta_4\}_t-3(\beta_3+3\beta_4)\gamma_{1,t}-2\gamma_1(\beta_3+2\beta_4)h_2-\frac13(2\gamma_1h_2+5\gamma_2h_1)+\frac{10}{3}(\beta_5-\beta_6)\gamma_1^2$ \\
$\qquad -\frac13(19\beta_6+2\beta_5)\gamma_1\gamma_2+4(2\beta_5+\beta_6)\gamma_2^2+(\{22\beta_5-\beta_6\}\beta_3+\{8\beta_5-29\beta_6\}\beta_4)\gamma_1$\\
$\delta_7=-\gamma_{1,tt}+(2h_1+h_2+7\beta_6-\beta_5)\gamma_{1,t}+(3\beta_6h_2-2\{2\beta_6+\beta_5\}_t)\gamma_1 +6\beta_7(\gamma_1-6\gamma_2)\gamma_2$\\
$\qquad +(\beta_6\{11\beta_6-7\beta_5\}+6\beta_7\{\beta_4-\beta_3\}-4\beta_5^2)\gamma_1$\\
$\delta_8=\frac53(2\beta_3+\beta_4)(3\{\beta_3-\beta_4\}-\gamma_2)\gamma_{1,t}+\frac23(\beta_3-\beta_4)(3\{2\beta_3+\beta_4\}+\gamma_1)\gamma_1h_2+
\frac13(2\gamma_1+5\gamma_2$\\
$\qquad -3\beta_3+21\beta_4)\gamma_1\beta_{4,t}+\frac23(5\gamma_2-\gamma_1-21\beta_3+12\beta_4)\gamma_1\beta_{3,t}+ \frac23(2\beta_5+\beta_6)(6\beta_4+12\beta_3+\gamma_1)\gamma_2^2$\\
$\qquad -\frac43(\beta_5-\beta_6)(3\{\beta_3-\beta_4\}-\gamma_2)\gamma_1^2
+4(\{4\beta_3-\beta_4\}\beta_5-\{\beta_3+2\beta_4\}\beta_6)\gamma_1\gamma_2$\\
$\qquad+18(\beta_4-\beta_3)(2\beta_5\beta_3-\{\beta_3+\beta_4\}\beta_6)\gamma_1$\\
$\delta_9=l(\beta_3,\beta_4)+(2\{7\beta_5 -4\beta_6\}\beta_{3,t}+\{\beta_5-7\beta_6\}\beta_{4,t})\gamma_1 +2\{\frac13(\gamma_1-5\gamma_2)+7\beta_3-4\beta_4\}\beta_{5,t}$\\
$\qquad +5\left\{\frac13(2\beta_5+\beta_6\gamma_2)+(\beta_3+2\beta_4)\beta_6+(4\beta_3-\beta_4)\beta_5\right\}\gamma_{1,t}-\{\frac13(2\gamma_1+5\gamma_2)-\beta_3+7\beta_4\}\beta_{6,t}
$\\
$\qquad +2\left\{\frac13(\beta_5-\beta_6)\gamma_1-(\beta_3+2\beta_4)\beta_6+(4\beta_3-\beta_4)\beta_5\right\}\gamma_1h_2+2(\beta_5-\beta_6)^2\gamma_1^2$\\
$\qquad -2(\{2\beta_5+\beta_6\}^2 +\{6(\beta_4+2\beta_3)+\gamma_1\}\beta_7)\gamma_2^2 +4(\{\beta_6-\beta_5\}\{\beta_6+2\beta_5\}+3\{\beta_4-\beta_3\}\beta_7)\gamma_1\gamma_2$\\
$\qquad +18(\{\beta_6^2-\beta_5^2\}\beta_4+\{\beta_3-\beta_4\}^2\beta_7
+2\{\beta_5-\beta_6\}\beta_3\beta_5)\gamma_1$\\
$\delta_{10}=l(\beta_5,\beta_6)-5\beta_7\gamma_2\gamma_{1,t} +\{2\beta_5^2-\beta_6^2-\beta_5\beta_6+3(\beta_3-\beta_4)\beta_7\}\{5h_1+2h_2\}\gamma_1+\{6(\beta_4-\beta_3)_t\beta_7$\\
$\qquad +2(4\beta_6-7\beta_5)\beta_{5,t}+(7\beta_6-\beta_5)\beta_{6,t}+5(\gamma_2-3[\beta_3-\beta_4])\beta_{7,t}\}\gamma_1+12(2\beta_5+\beta_6)\beta_7\gamma_2^2$\\
$\qquad-6(\beta_5-\beta_6)^2(2\beta_5+\beta_6)\gamma_1 +12(\beta_5-\beta_6)(3\{\beta_4-\beta_3\}+\gamma_2)\beta_7\gamma_1$\\
$\delta_{11}=3\beta_{7,tt}\gamma_1-3\beta_7\gamma_{1,tt}+3(2h_1^2\gamma_1+h_2)\beta_7-3(2\gamma_{1,t}+\gamma_1h_2)\beta_{7,t}+6\{\beta_5-\beta_6\}_t\beta_7\gamma_1$\\
$\qquad +\{\beta_6-\beta_5\}\{15(\beta_7\gamma_{1,t}-\beta_{7,t}\gamma_1)+6\beta_7\gamma_1h_2\}+18(\{\beta_5-\beta_6\}^2\gamma_1-\beta_7\gamma_2)\beta_7$ \\
        \hline
    \end{tabular}
    \caption{Coefficients of equation \eqref{gen_DH_ODE}. Designations: $h_1=\{\log \gamma_1\}_t$, $h_2=\{\log \gamma_2\}_t$, $l(f,g)=\{2f+g\}_{tt}\gamma_1+\{(2h_1+h_2)\gamma_1h_1-\gamma_{1,tt}\}\{2f+g\}-\{2f+g\}_t\{2h_1+h_2\}\gamma_1$.} \label{table:DS_DH_coeff}
\end{table}

In the case of the Euler's dynamical system \eqref{Euler_MP} equation \eqref{gen_DH_ODE} takes the form
\begin{equation}
\begin{gathered}
\label{gen_DH_ODE_partial case_linear_Euler}
w_{ttt}-(\beta_1 w+6\beta_6)w_{tt}-2\beta_1 w_t^2-(2\beta_1^2w^2-7\beta_1\beta_6w+14\beta_7\beta_1-11\beta_6^2)w_t
+2\beta_6\beta_1^2w^3\\
-2\beta_1(2\beta_6^2+\beta_7\beta_1)w^2
+6\beta_6(4\beta_7\beta_1-\beta_6^2)w+18\beta_7(\beta_6^2-\beta_7\beta_1)=0.
\end{gathered}
\end{equation}
The function $y=-\beta_1w$ satisfies the Chazy-VII equation
\begin{equation}
\begin{gathered}
\label{gen_DH_ODE_Chazy7}
y_{ttt}+yy_{tt}+2y_t^2-2y^2y_t=0
\end{gathered}
\end{equation}
provided that the parameters $\beta_6$, $\beta_7$ vanish. The general solution of the Chazy-VII equation is known: $y=\{\log\wp_{t}(t-t_0;g_2,g_3)\}_t$ \cite{Cosgrove01}. Here
$\wp(t-t_0;g_2,g_3)$ is the Weierstrass elliptic function and $t_0$, $g_2$, $g_3$ are arbitrary constants.

Suppose that the present status concerning integrability of equation \eqref{gen_DH_ODE} at particular choice of the parameters  is unknown; then one comes across a problem of finding some exact solutions. Let us derive certain  well--behaved (meromorphic) exact solutions of equation \eqref{gen_DH_ODE}.  For this aim we shall use the method based on Mittag--Leffler's expansions of meromorphic functions, see also \cite{Demina01, Demina02, Demina03, Demina04}. In autonomous case equation \eqref{gen_DH_ODE} is invariant under the transformation $t\mapsto t-t_0$, where $t_0$ is an arbitrary parameter. Thus equation \eqref{gen_DH_ODE} may admit periodic meromorphic solutions.

A single-valued function $f(t)$ of one complex variable $t$ is called meromorphic (in the complex plane)
if it does not possess singularities in the finite points other than (isolated) poles. A meromorphic function holomorphic in the whole complex plane is entire.  Any meromorphic functions is characterized by its behavior at infinity. There are three possibilities: (a) a point at infinity is a removable singularity or an (isolated) pole; (b) a point at infinity infinity is an (isolated) essential singularity; (c) a point at infinity infinity is a non--isolated singularity.
In the case (a) the function $f(t)$
is rational. In the case (b) the function $f(t)$ is  transcendental entire. Finally, in the remaining case (c) the function $f(t)$ possesses infinitely many poles and is called transcendental meromorphic.

The set of periods of a nonconstant meromorphic function $f(t)$ is a discrete additive group: $\Omega=\mathbb{Z}$ $\tau_1$ $+$ $\mathbb{Z}$ $\tau_2$. The parameters $\tau_1$, $\tau_2$ are called fundamental (or principal) periods. If $\tau_1\tau_2\neq0$, then the meromorphic function is said to be doubly periodic (or elliptic). Note that the ratio $\tau_1/\tau_2$ cannot be real. If one of the fundamental periods, for example, $\tau_2$ is zero, then the meromorphic function is said to be simply periodic.

All the values an elliptic function takes in the parallelogram built on the periods $\tau_1$, $\tau_2$ periodically arise in other point of the complex plane. The number of poles in a parallelogram of
periods, counting multiplicity, is called the order of an elliptic function.  Analogously, behavior of a simply periodic meromorphic function in the complex plane is characterized by its behavior in a stripe built on the period $\tau_1$.

In what follows the fundamental period of a simply periodic meromorphic function we denote as $T$. Let us construct simply periodic meromorphic functions with one simple pole in a stripe of periods. Without loss of generality we place this pole at the origin. Using Mittag--Leffler's expansions of meromorphic functions, we get
\begin{equation}
\begin{gathered}
\label{Periodic_example} f(t)=\frac{c_{-1}}{t}+\sum_{n\in\mathbb{Z},\,n\neq
0}^{}\left[\frac{c_{-1}}{t-nT}+\frac{c_{-1}}{nT}\right]+ h(t),
\end{gathered}
\end{equation}
where $h(t)$ is a periodic  entire function with the fundamental period $T$. Note that this function can be presented in the form
\begin{equation}
\begin{gathered}
\label{Periodic_example_h}
h(t)=\sum_{k=-\infty}^{\infty}h_k\exp\left[\frac{2\pi i k }{T}t\right].
\end{gathered}
\end{equation}
The series in expression \eqref{Periodic_example} is uniformly convergent in any domain
not including the points $t=\mathbb{Z}T$. The constant term is added to each item in order to provide convergence. Implementing the summation yields
\begin{equation}
\begin{gathered}
\label{Periodic_example_res} f(t)=\frac{\pi}{T}\cot
\left(\frac{\pi t}{T}\right)+ h(t),
\end{gathered}
\end{equation}
Calculating the derivative of expression \eqref{Periodic_example} we find a simply periodic meromorphic function possessing (in a stripe of periods) one double pole with zero residue.  In similar way we can construct simply periodic meromorphic functions with triple, etc pole in a stripe of periods. Any simply periodic meromorphic function with finite amount of poles $b_1$, $\ldots$, $b_M$ in the fundamental stripe of periods  is given by
\begin{equation}
\begin{gathered}
\label{Periodic_example_many poles} f(t)= \frac{\pi}{T}\left\{\sum_{i=1}^{M}\sum_{k=1}^{p_i}\frac{(-1)^{k-1}
c_{-k}^{(i)}}{(k-1)!}\frac{d^{k-1}}{dt^{k-1}}\right\}\cot
\left(\frac{\pi \{t-b_i\}}{T}\right)+ h(t),
\end{gathered}
\end{equation}
where $p_i$ is the order of the pole $z=b_i$ and $h(t)$ is again a periodic entire function.

The simplest example of an elliptic function is the Weierstrass $\wp$--function defined as
\begin{equation}
\begin{gathered}
\label{Periodic_Elliptic_Weierstrass}
\wp(t)=\frac{1}{t^2}+{\sum_{(n,\,m)\neq(0,\,0)}}
\left[\frac{1}{(t-n\tau_1-m\tau_2)^2}-\frac{1}{(n\tau_1+m\tau_2)^2}\right].
\end{gathered}
\end{equation}
The series in expression \eqref{Periodic_Elliptic_Weierstrass} is uniformly convergent in any domain
not including the points $t=\mathbb{Z}\tau_1+\mathbb{Z}\tau_2$. The Weierstrass elliptic function satisfies the following first--order
ordinary differential equation
\begin{equation}
\begin{gathered}
\label{Wier} (\wp_t)^2=4\wp^3-g_2\wp-g_3.
\end{gathered}
\end{equation}
Integrating the expression in \eqref{Periodic_Elliptic_Weierstrass}, multiplying the result by $-1$ and using normalizing condition
\begin{equation}
\begin{gathered}
\label{Zeta_NC}\lim_{t\rightarrow0}\left[\int\wp(t) dt +\frac{1}{t}\right]=0.
\end{gathered}
\end{equation}
we obtain the Weierstrass $\zeta$--function
\begin{equation}
\begin{gathered}
\label{Periodic_Elliptic_Weierstrass_zeta}
\zeta(t)=\frac{1}{t}+{\sum_{(n,\,m)\neq(0,\,0)}}
\left[\frac{1}{t-n\tau_1-m\tau_2}+\frac{1}{n\tau_1+m\tau_2}+\frac{t}{(n\tau_1+m\tau_2)^2}\right].
\end{gathered}
\end{equation}
Note that the path of integration in \eqref{Zeta_NC} avoids the poles of $\wp(t)$. Thus we see that $\zeta_t(t)=-\wp(t)$. The Weierstrass $\zeta$--function is not elliptic since the sum of the residues of any elliptic function at its poles in a fundamental parallelogram should be zero. But in combinations
\begin{equation}
\begin{gathered}
\label{PDE1_ES_General_zeta}
f(t)=\sum_{i=1}^{M}\varkappa_i \zeta(t-b_i)+h_0
\end{gathered}
\end{equation}
such that the condition $\varkappa_1+\ldots +\varkappa_M =0$ is satisfied the Weierstrass $\zeta$--function produces an elliptic function with $M$ simple poles in a parallelogram of periods: $b_1$, $\ldots$, $b_M$.

Using expression \eqref{PDE1_ES_General_zeta} and the relation $\wp(t)=-\zeta_t(t)$ we can construct the general expression for an elliptic function with $M$ poles in a parallelogram of periods. It is
\begin{equation}
\begin{gathered}
\label{PDE1_ES_General}
f(t)=\left\{\sum_{i=1}^{M} \sum_{k=1}^{p_i}\frac{(-1)^{k-1}
c_{-k}^{(i)}}{(k-1)!}\frac{d^{k-1}}{dt^{k-1}}\right\}\zeta(t-b_i)+h_0,\quad \sum_{i=1}^{M}c_{-1}^{(i)}=0.
\end{gathered}
\end{equation}

 Let us consider a nonlinear ordinary differential
equation
\begin{equation}
\label{MP_ODE} K[w(t)]=0,
\end{equation}
where the expression $K[w(t)]$ is a polynomial in $w(t)$ and its derivatives. Note that equation \eqref{gen_DH_ODE} in autonomous case is of the form \eqref{MP_ODE}. Behavior of a meromorphic function in  a neighborhood of a pole can be characterized by means of its  Laurent series. If one wishes to find its elliptic \eqref{PDE1_ES_General} or simply periodic \eqref{Periodic_example_many poles} solutions explicitly it is necessary to construct families of  Laurent series satisfying equation \eqref{MP_ODE}. For this aim the  Painlev\'{e} methods can be used. Taking equation \eqref{gen_DH_ODE} and applying these methods one can get
\begin{equation}
\begin{gathered}
\label{LE_PDE1}
w^{(j)}(t)=\sum_{k=1}^{p_j}\frac{c_{-k}^{(j)}}{\left\{t-t_0\right\}^k}+\sum_{k=0}^{\infty}c_k^{(j)}\left\{t-t_0\right\}^k,\quad t\rightarrow t_0,\quad  j\in \mathbb{J}.
\end{gathered}
\end{equation}
Here $p_i$ is the order of the pole $t=t_0$ and the index $j$ numerates different families of admissible series. The coefficients of the series can be sequently calculated, some of them may turn to be arbitrary. Note that particular series may exist under certain restrictions on the parameters of the original equation.

Simply periodic solutions with finite amount of poles in a stripe of periods can be
presented as follows
\begin{equation}
\begin{gathered}
\label{PDE1_SPS_General}
w(t)=\frac{\pi}{T}\sum_{j\subseteq \mathbb{J}}\left\{\sum_{i=1}^{M_j}\sum_{k=1}^{p_j}\frac{(-1)^{k-1}
c_{-k}^{(j)}}{(k-1)!}\frac{d^{k-1}}{dt^{k-1}}\right\}\cot
\left[\frac{\pi \left\{t-b_i^{(j)}\right\}}{T}\right] +h(t),
\end{gathered}
\end{equation}
where $T$ is the fundamental period, $M_j$ is the number of poles of type  $j\subseteq \mathbb{J}$, and $h(t)$ is a periodic entire function, see \eqref{Periodic_example_h}.
The general expression for elliptic solutions is the following
\begin{equation}
\begin{gathered}
\label{PDE1_ES_General}
w(t)=\sum_{j\subseteq \mathbb{J}}\left\{\sum_{i=1}^{M_j} \sum_{k=1}^{p_j}\frac{(-1)^{k-1}
c_{-k}^{(j)}}{(k-1)!}\frac{d^{k-1}}{dt^{k-1}}\right\}\zeta(t-b_i^{(j)})+h_0.
\end{gathered}
\end{equation}
The theorem on total sum of the residues of any elliptic function gives a necessary condition for an elliptic solution to exist
\begin{equation}
\begin{gathered}
\label{PDE1_ES_General_Necessary_Condition}
\sum_{j\subseteq \mathbb{J}}\sum_{i=1}^{M_j} c_{-1}^{(j)}=0.
\end{gathered}
\end{equation}
Note that we omit arbitrary constant $t_0$ resulting from the invariance of equation \eqref{MP_ODE} under the transformation $t\mapsto t-t_0$.

Further let us discuss the problem of finding solutions \eqref{PDE1_SPS_General}, \eqref{PDE1_ES_General} in explicit form. In this article  we shall set $h(t)=h_0$ in the case of simply periodic solutions \eqref{PDE1_SPS_General}, the possibility $h(t)\neq h_0$ requires additional treatment.
First of all let us mention that partial derivatives of relations  \eqref{PDE1_ES_General} are  polynomials in $\cot
\left[\pi \left(t-b_i^{(j)}\right)/T\right]$ and elliptic functions, accordingly.   Consequently, substituting relations  \eqref{PDE1_SPS_General}, \eqref{PDE1_ES_General} into  equation \eqref{MP_ODE}, we obtain  a polynomial in $\cot
\left[\pi \left(t-b_i^{(j)}\right)/T\right]$ and an elliptic function, respectively. If the resulting function
does no have poles, then from the Liouville theorem it immediately follows that such a function is a constant. Instead of substituting the solutions themselves we may substitute their Laurent series in a neighborhood of the poles $b_i^{(j)}$.

Our algorithm can be subdivided into several steps.

\textit{Step 1.} Perform local singularity analysis for solutions of equation \eqref{MP_ODE}. Construct all the  Laurent series of the form \eqref{LE_PDE1}.

\textit{Step 2.} Write down general expressions  \eqref{PDE1_SPS_General}, \eqref{PDE1_ES_General}.

\textit{Step 3.} Take any solutions of step $2$ and find its Laurent series in a neighborhood of the poles.

\textit{Step 4.} Substitute all those  Laurent  found at step $3$ that are captured by a supposed solution into the original equation and set to zero coefficients at negative and zero powers of the expression $\left\{t-b_i^{(j)}\right\}$.

\textit{Step 5.} Solve obtained algebraic system.

As a rule, it is not an easy task to find the positions of the poles $b_i^{(j)}$. It turns out that it is sufficient to  obtain the values $B_i^{(j)}=\cot
\left[lb_i^{(j)}\right]$, $l=\pi/T$ in the case of simply periodic solutions and the values $A_i^{(j)}=\wp(b_i^{(j)};g_2,g_3)$ and $B_i^{(j)}=\wp_t(b_i^{(j)};g_2,g_3)$ in the case of elliptic solutions instead. Note that the parameters $A_i^{(j)}$, $B_i^{(j)}$ (in the elliptic case) are not independent, they are related by the equalities
\begin{equation}
\begin{gathered}
\label{DS_elliptic_sols_add_eq}
\left\{B_i^{(j)}\right\}^2=4\left\{A_i^{(j)}\right\}^3-g_2A_i^{(j)}-g_3.
\end{gathered}
\end{equation}
Subsequently, addition formulae for periodic meromorphic functions can be used to rewrite solutions  \eqref{PDE1_SPS_General}, \eqref{PDE1_ES_General} in terms of the new parameters, see  \cite{Demina01, Demina02, Demina03, Demina04}.

As soon as a solution is found it is necessary to verify that its  Laurent series in a neighborhood of different poles in the principle parallelogram of periods are in fact distinct.

Now we pass on to finding well--behaved solutions of equation \eqref{gen_DH_ODE}. In detail we shall consider the case of elliptic solutions only. If at least one of the parameters $\delta_1$, $\delta_2$, $\delta_3$, $\delta_4$ is non--zero and some conditions on these parameters are satisfied, then equation \eqref{gen_DH_ODE} may have solutions with simple poles. Thus the simplest (second--order) elliptic solution is of the form
\begin{equation}
\begin{gathered}
\label{DS_elliptic_sols_gen}
w(t)=c_{-1}\zeta(t-t_0;g_2,g_3)-c_{-1}\zeta(t-b-t_0;g_2,g_3)+\tilde{h}_0,
\end{gathered}
\end{equation}
where  $g_2$, $g_3$, $c_{-1}$, $\tilde{h}_0$ are constants to be found,  and $t_0$ is an arbitrary constant resulting from autonomy of equation \eqref{gen_DH_ODE}. In what follows we omit this constant.  Function \eqref{DS_elliptic_sols_gen} possesses two simple poles ($t=0$ and $t=b$) in  the principle parallelogram of periods. According to the preceding remarks  we introduce the parameters: $A=\wp(b;g_2,g_3)$ and $B=\wp_t(b;g_2,g_3)$  related by the equality $B^2=4A^3-g_2A-g_3$. Using addition  formulae for the Weierstrass $\zeta$--function, we rewrite relation \eqref{DS_elliptic_sols_gen} as follows
\begin{equation}
\begin{gathered}
\label{DS_elliptic_sols_gen_res}
w(z)=-\frac{c_{-1}\{\wp_t(t;g_2,g_3)+B\}}{2\{\wp(t;g_2,g_3)-A\}}+h_0,\quad h_0=\tilde{h}_0+c_{-1}A.
\end{gathered}
\end{equation}

Let us turn to  finding solutions \eqref{DS_elliptic_sols_gen_res} in explicit form. Substituting elliptic function \eqref{DS_elliptic_sols_gen} into equation \eqref{gen_DH_ODE}, we get another elliptic function. If the resulting function
does no have poles, it immediately follows that such a function is a constant since an elliptic function without poles does not exist. Instead of substituting the solution itself we shall substitute its Laurent series in a neighborhood of the poles $t=0$, $t=b$. These Laurent series take the form
\begin{equation}
\begin{gathered}
\label{DS_elliptic_sols_LS}
w(t)=\frac{c_{-1}}{t}+h_0+c_{-1}At-\frac{c_{-1}B}{2}t^2+c_{-1}\left(A^2-\frac{g_2}{10}\right)t^3+\ldots,\quad t\rightarrow 0\\
w(t)=-\frac{c_{-1}}{t-b}+h_0-c_{-1}A\{t-b\}-\frac{c_{-1}B}{2}\{t-b\}^2\hfill \\
-c_{-1}\left(A^2-\frac{g_2}{10}\right)\{t-b\}^3+\ldots,\quad
 t\rightarrow b.
\end{gathered}
\end{equation}
In order to calculate these series we have used expression \eqref{DS_elliptic_sols_gen}. Substituting  series \eqref{DS_elliptic_sols_LS} into equation \eqref{gen_DH_ODE} we set to zero coefficients at negative and zero powers of $t$ and $t-b$. We need only five first coefficients from
each of the series. The resulting system contains nine equations including  $B^2=4A^3-g_2A-g_3$. For all the solutions we find
\begin{equation}
\begin{gathered}
\label{DS_elliptic_sols_par1}
c_{-1}=\varepsilon\sqrt{-\frac{6}{\delta_3}},\, h_0=-\frac{\delta_6}{2\delta_3},\, g_3=4A^3-g_2A-B^2,\,
 A=\frac{4\delta_3\delta_7-\delta_6^2}{24\delta_3},\, \varepsilon=\pm 1
\end{gathered}
\end{equation}
and the following restrictions on the parameters of the original equation
\begin{equation}
\begin{gathered}
\label{DS_elliptic_sols_par2}
\delta_9=\frac{(4\delta_3\delta_7-\delta_6^2)\delta_2\delta_3+12(\delta_6\delta_8+2\delta_4\delta_7)\delta_3-18\delta_4\delta_6^2}{8\delta_3^2},\hfill \\
\delta_5=\frac{12\delta_8\delta_3-18\delta_4\delta_6-\delta_2\delta_3\delta_6}{4\delta_3^2},\quad \delta_1=\frac{3\delta_4}{\delta_3}-\frac{\delta_2}{2},\quad \delta_3\neq 0.
\end{gathered}
\end{equation}
The first family exists under additional restrictions: $\delta_2=0$, $\delta_4=0$, $\delta_8=0$, $\delta_{10}=0$, $\delta_{11}=0$. In this case the parameters $g_2$, $B$ are arbitrary and  solution \eqref{DS_elliptic_sols_gen_res} is the general solution of the equation in question.

For the second family we get $\delta_2=0$, $\delta_4=0$, $\delta_{10}=(3\delta_8\delta_7)/\delta_3$, $\delta_8\neq0$. The parameter $g_2$ is arbitrary and the parameter $B$ is given by
\begin{equation}
\begin{gathered}
\label{DS_elliptic_sols_par3}
B=\varepsilon\sqrt{-\frac{6}{\delta_3}}\frac{\delta_8\delta_6^3+4\delta_{11}\delta_3^3-2\delta_{10}\delta_6\delta_3^2}{144\delta_3\delta_8}.
\end{gathered}
\end{equation}
The third family exists provided that $\delta_2=0$, $\delta_4\neq 0$ and
\begin{equation}
\begin{gathered}
\label{DS_elliptic_sols_par4}
\delta_{11}=\frac{(2\delta_8\delta_3-3\delta_4\delta_6)(2\delta_{10}\delta_3^2-6\delta_8\delta_7\delta_3+9\delta_7\delta_6\delta_4)}{4\delta_3^3\delta_4}.
\end{gathered}
\end{equation}
 The parameter $g_2$ is arbitrary and the parameter $B$ takes the form
\begin{equation}
\begin{gathered}
\label{DS_elliptic_sols_par5}
B=\varepsilon\sqrt{-\frac{6}{\delta_3}}\frac{(4\delta_{10}\delta_3-12\delta_7\delta_8)\delta_3^2+(\delta_6^2+12\delta_7\delta_3)\delta_4\delta_6}{144\delta_3\delta_4}.
\end{gathered}
\end{equation}
If the following conditions
\begin{equation}
\begin{gathered}
\label{DS_elliptic_sols_par6}
\delta_{4}=-\frac{\delta_2\delta_3}{2},\quad \delta_{10}=\frac{(2\delta_2\delta_6+3\delta_8)\delta_7}{\delta_3},\quad \delta_2\neq0
\end{gathered}
\end{equation}
are valid, then we obtain the fourth family of solutions. The parameter $B$ is an arbitrary constant and the parameter $g_2$ reads as
\begin{equation*}
\begin{gathered}
\label{DS_elliptic_sols_par7}
g_2=\left\{\delta_6+\frac{\delta_8}{\delta_2}\right\}B+\varepsilon\sqrt{-\frac{6}{\delta_3}}\frac{(6\delta_7\delta_3-\delta_6^2)\delta_8
\delta_6-2(\delta_2\delta_7^2+2\delta_{11}\delta_3)\delta_3^2+(6\delta_7\delta_3-\delta_6^2)\delta_6^2\delta_2}{144\delta_2\delta_3}
\end{gathered}
\end{equation*}
Finally, under the conditions $\delta_2\neq0$, $\delta_{4}\neq-(\delta_2\delta_3)/2$ equation \eqref{gen_DH_ODE} possesses the fifth family of solutions with the parameters $B$ and $g_2$ given by
\begin{equation}
\begin{gathered}
\label{DS_elliptic_sols_par8}
B=\varepsilon\sqrt{-\frac{6}{\delta_3}}\frac{(\delta_2\delta_3+2\delta_4)\delta_6^3-4(\delta_2\delta_3-6\delta_4)
\delta_3\delta_6\delta_7+8(\delta_{10}\delta_3-3\delta_7\delta_8)\delta_3^2}{144(\delta_2\delta_3+2\delta_4)\delta_3},\\
g_2=\frac{(6\delta_4-\delta_2\delta_3)(18\delta_4+\delta_2\delta_3)\delta_7\delta_6^2+4\{\mu_2\delta_3-\mu_1\delta_6\}\delta_3}
{48(\delta_2\delta_3+2\delta_4)\delta_3^2\delta_2},\hfill
\end{gathered}
\end{equation}
where we have used notation
\begin{equation}
\begin{gathered}
\label{DS_elliptic_sols_par9}
\mu_1=(\delta_2\delta_3-6\delta_4)\delta_3\delta_{10}
+2(18\delta_4-\delta_2\delta_3)\delta_7\delta_8,\hfill \\
\mu_2=(\delta_2\delta_3+2\delta_4)(2\delta_3\delta_{11}+\delta_2\delta_7^2)+4(3\delta_8\delta_7-\delta_3\delta_{10})\delta_8
\end{gathered}
\end{equation}

Further, let us note that if  $\delta_3=0$, $\delta_4=0$ and certain conditions on the parameters $\delta_1$, $\delta_2$, $\delta_6$, $\delta_8$  are satisfied, then  equation \eqref{gen_DH_ODE} may have solutions with second--order poles. Second--order elliptic functions having one double pole in the fundamental parallelogram of periods are the following
\begin{equation}
\begin{gathered}
\label{DS_elliptic_sols_gen2}
w(z)=c_{-2}\wp(t-t_0;g_2,g_3)+h_0,
\end{gathered}
\end{equation}
Without loss of generality, let us omit the arbitrary constant $t_0$. Finding the Laurent series of this function in a neighborhood of its pole $t=0$, we obtain
\begin{equation}
\begin{gathered}
\label{DS_elliptic_sols_gen2_LS}
w(z)=\frac{c_{-2}}{t^2}+h_0+\frac{c_{-2}}{20}g_2t^2+\frac{c_{-2}}{28}g_3t^4+\ldots,\quad  t \rightarrow 0
\end{gathered}
\end{equation}
Here we have written down all the coefficients  essential for further calculations. Substituting series \eqref{DS_elliptic_sols_gen2_LS} into equation \eqref{gen_DH_ODE} with $\delta_3=0$, $\delta_4=0$ and setting to zero the coefficients at negative and zero powers of $t$ yields six families of solutions. For all the families the parameters $c_{-2}$, $h_0$ are following
\begin{equation}
\begin{gathered}
\label{DS_elliptic_sols_gen2_par1}
c_{-2}=-\frac{12}{\alpha_6},\quad h_0=-\frac{\alpha_7}{\alpha_6},\quad \alpha_6\neq0.
\end{gathered}
\end{equation}
Along with this we find two restrictions on the parameters of the original equation
\begin{equation}
\begin{gathered}
\label{DS_elliptic_sols_gen2_par2}
\alpha_8=\frac{\alpha_6}{6}(2\alpha_2+3\alpha_1),\quad \alpha_9=\frac12\alpha_5\alpha_6+(\alpha_1+\alpha_2)\alpha_7.
\end{gathered}
\end{equation}
The first family of solutions exists under additional constrains $\alpha_2\neq0$, $\alpha_1\neq -2\alpha_2$. The invariants $g_2$, $g_3$ are given by
\begin{equation}
\begin{gathered}
\label{DS_elliptic_sols_gen2_par3}
g_2=\frac{\alpha_7^2}{12}+\frac{(\alpha_5\alpha_7-\alpha_{10})\alpha_6}{6(\alpha_1+2\alpha_2)},\quad
g_3=\frac{\alpha_7^3}{216}+\frac{\alpha_6^2\alpha_{11}}{144\alpha_2}+
\frac{(\alpha_5\alpha_6+2\alpha_2\alpha_7)(\alpha_5\alpha_7-\alpha_{10})\alpha_6}{144(\alpha_1+2\alpha_2)\alpha_2}
\end{gathered}
\end{equation}
If $\alpha_2=0$, $\alpha_1\neq0$, and $\alpha_{11}=(\alpha_{10}-\alpha_5\alpha_7)\alpha_5/\alpha_1$ we obtain the second family of elliptic solutions with an arbitrary parameter $g_3$ and the following value of the parameter $g_2$:
\begin{equation}
\begin{gathered}
\label{DS_elliptic_sols_gen2_par4}
g_2=\frac{\alpha_7^2}{12}+\frac{(\alpha_5\alpha_7-\alpha_{10})\alpha_6}{6\alpha_1}.
\end{gathered}
\end{equation}
The third family of solutions exists whenever $\alpha_1=-2\alpha_2$, $\alpha_2\neq0$, $\alpha_{10}=\alpha_5\alpha_7$. In this case the parameter $g_2$ is arbitrary and the parameter $g_3$ is given by
\begin{equation}
\begin{gathered}
\label{DS_elliptic_sols_gen2_par5}
g_3=\frac{1}{12}\left(\alpha_7+\frac{\alpha_5\alpha_6}{2\alpha_2}\right)g_2-\frac{\alpha_7^3}{432}+\frac{(2\alpha_6\alpha_{11}-\alpha_5\alpha_7^2)\alpha_6}{288\alpha_2}.
\end{gathered}
\end{equation}
Further, if $\alpha_1=0$, $\alpha_2=0$, $\alpha_{10}=\alpha_5\alpha_7$, $\alpha_5\neq0$, then we obtain the fourth family of solutions with an arbitrary parameter $g_3$ and the following value of the parameter $g_2$:
\begin{equation}
\begin{gathered}
\label{DS_elliptic_sols_gen2_par6}
g_2=\frac{\alpha_7^2}{12}-\frac{\alpha_6\alpha_{11}}{6\alpha_5}.
\end{gathered}
\end{equation}
Finally, under the restrictions $\alpha_1=0$, $\alpha_2=0$, $\alpha_{5}=0$, $\alpha_{10}=0$, $\alpha_{11}=0$  expression \eqref{DS_elliptic_sols_gen2} is the general solution of the corresponding equation. In this case the invariants $g_2$, $g_3$ are arbitrary, the third arbitrary constant is the parameter $t_0$.

According to results of articles \cite{Demina01,Demina02,Demina03,Demina04} we have classified all second--order elliptic solutions of equation \eqref{gen_DH_ODE}. Note that in the case $g_2^3-27g_3^2=0$ elliptic solutions degenerate to simply periodic or rational. On use of obtained families of elliptic solutions one may construct exact solutions of the corresponding dynamical systems.





\section {Conclusion}

In this article we have considered  polynomial multi--particle dynamical systems in the plane. We have presented a method, the polynomial method, which replaces integration of  polynomial multi--particle dynamical systems by constructing polynomial solutions of partial differential equations.

With the help of the polynomial method we have studied several interesting dynamical systems possessing equilibria given by the roots of classical and some other families of orthogonal polynomials. As a by--product of our results we have obtained several new families of orthogonal polynomials.

We have integrated a wide class of two and three--particle polynomial dynamical systems including a number of physically relevant systems, such as, the Euler's system,  the Darboux--Halphen system, and their generalizations.

\section {Acknowledgements}

This research was partially supported by Russian Science Foundation, project to support research carried out by individual research groups No. 14-11-00258

\end{document}